\documentclass[sigconf]{acmart}

\newif\iflong
\longtrue 

\usepackage[vlined,noend,linesnumbered,noresetcount]{algorithm2e}
\usepackage{balance}

\newtheorem{definition}{Definition}
\newtheorem{example}{Example}
\newtheorem{theorem}{Theorem}
\newtheorem{lemma}{Lemma}
\newtheorem{proposition}{Proposition}
\newtheorem{myclaim}{\sc Claim}

\newcommand{\tr}[2]{\iflong{}\S#1\else{}\cite[\S#2]{this-extended-version}\fi}
\newcommand{\safetydefs}{A}
\newcommand{\regsafety}{B}
\newcommand{\connected}{C}

\makeatletter
\newcommand{\removelatexerror}{\let\@latex@error\@gobble}
\makeatother

\SetKwProg{Function}{function}{:}{}
\SetKwBlock{SubAlgoBlock}{}{end}
\newcommand{\SubAlgo}[2]{#1 \SubAlgoBlock{#2}}
\newcommand{\RG}[2]{{#1}\setminus{#2}}

\newcommand{\onreceive}{\textbf{when received}\xspace}
\newcommand{\send}{\textbf{send}\xspace}
\newcommand{\from}{\textbf{from}\xspace}
\newcommand{\ToAll}{\textbf{to all}\xspace}
\newcommand{\To}{\textbf{to}\xspace}
\newcommand{\precond}{\textbf{pre:}\xspace}

\newcommand{\assign}[2]{\ensuremath{#1} \ensuremath{\leftarrow} \ensuremath{#2}}
\newcommand{\return}{\textbf{return}\xspace}

\newcommand{\waituntil}{\textbf{wait until received}\xspace}

\newcommand{\upd}{u}
\newcommand{\ter}[1]{U_{#1}}

\newcommand{\VV}{\mathcal{V}}
\newcommand{\CC}{\mathcal{C}}
\newcommand{\GG}{\mathcal{G}}
\newcommand{\PP}{\mathcal{P}}
\newcommand{\FS}{\mathcal{F}}

\newcommand{\RR}{\mathcal{R}}
\newcommand{\WW}{\mathcal{W}}
\newcommand{\ST}{\mathcal{S}}

\newcommand{\GST}{{\sf GST}}
\newcommand{\ALG}{\mathcal{A}}

\newcommand{\readquorum}{\mathsf{quorum\_get}}
\newcommand{\writequorum}{\mathsf{quorum\_set}}

\newcommand{\GETREQ}{{\tt GET\_REQ}}
\newcommand{\GETRESP}{{\tt GET\_RESP}}
\newcommand{\SETREQ}{{\tt SET\_REQ}}
\newcommand{\SETRESP}{{\tt SET\_RESP}}

\newcommand{\pset}{p_{\rm set}}
\newcommand{\pget}{p_{\rm get}}

\newcommand{\cset}{c_{\rm set}}
\newcommand{\cget}{c_{\rm get}}
\newcommand{\Rset}{R_{\rm set}}
\newcommand{\Rget}{R_{\rm get}}
\newcommand{\Wset}{W_{\rm set}}
\newcommand{\Wget}{W_{\rm get}}

\newcommand{\State}{\mathsf{state}}

\newcommand{\phase}{\mathsf{phase}}
\newcommand{\val}{\mathsf{val}}
\newcommand{\ts}{\mathsf{ver}}
\newcommand{\seq}{\mathsf{seq}}
\newcommand{\clock}{\mathsf{clock}}

\newcommand{\oread}{\mathsf{read}}
\newcommand{\owrite}{\mathsf{write}}

\newcommand{\CLOCKREQ}{{\tt CLOCK\_REQ}}
\newcommand{\CLOCKRESP}{{\tt CLOCK\_RESP}}

\newcommand{\pval}{\mathsf{my\_val}}
\newcommand{\view}{\mathsf{view}}
\newcommand{\cview}{\mathsf{aview}}
\newcommand{\leader}{\mathsf{leader}}
\newcommand{\viewtimer}{\mathsf{view\_timer}}
\newcommand{\starttimer}{\mathsf{start\_timer}}

\newcommand{\opropose}{\mathsf{propose}}

\newcommand{\entered}{\textsc{enter}}
\newcommand{\proposed}{\textsc{propose}}
\newcommand{\accepted}{\textsc{accept}}
\newcommand{\decided}{\textsc{decide}}

\newcommand{\OA}{{\tt 1A}}
\newcommand{\OB}{{\tt 1B}}
\newcommand{\TA}{{\tt 2A}}
\newcommand{\TB}{{\tt 2B}}

\newcommand{\rt}{{\sf rt}}
\newcommand{\WR}{{\sf wr}}
\newcommand{\ww}{{\sf ww}}
\newcommand{\rw}{{\sf rw}}

\newcommand{\propose}{\mathsf{propose}}

\newcommand{\Val}{\mathsf{Value}}
\newcommand{\ack}{\mathit{ack}}

\author{Alejandro Naser-Pastoriza} 
\affiliation{
  \institution{IMDEA Software Institute} 
  \institution{Universidad Polit\'ecnica de Madrid}
  \city{Madrid} 
  \country{Spain}
}

\author{Gregory Chockler}
\affiliation{
  \institution{University of Surrey} 
  \city{Guildford} 
  \country{United Kingdom}
}

\author{Alexey Gotsman} 
\affiliation{
  \institution{IMDEA Software Institute} 
  \city{Madrid} 
  \country{Spain}
}

\author{Fedor Ryabinin} 
\affiliation{
  \institution{IMDEA Software Institute} 
  \institution{Universidad Polit\'ecnica de Madrid}
  \city{Madrid} 
  \country{Spain}
}

\begin{document}

\acmYear{2025}
\copyrightyear{2025}
\acmConference[PODC '25]{ACM Symposium on Principles of Distributed Computing}{June 16--20, 2025}{Huatulco, Mexico} 
\acmBooktitle{ACM Symposium on Principles of Distributed Computing (PODC '25), June 16--20, 2025, Huatulco, Mexico} 
\acmDOI{10.1145/3732772.3733529}
\acmISBN{979-8-4007-1885-4/25/06}

\ccsdesc[500]{Theory of computation~Distributed algorithms}

\keywords{Distributed algorithms, lower bounds, atomic registers, atomic snapshots, lattice agreement, consensus}

\title
[Tight Bounds on Channel Reliability via Generalized Quorum Systems]
{Tight Bounds on Channel Reliability via\\ Generalized Quorum Systems}

\begin{abstract}
Communication channel failures are a major concern for the developers of modern
fault-tolerant systems. However, while tight bounds for process failures are
well-established, extending them to include channel failures has remained an
open problem. We introduce \emph{generalized quorum systems} -- a framework that
characterizes the necessary and sufficient conditions for implementing atomic
registers, atomic snapshots, lattice agreement and consensus under arbitrary
patterns of process-channel failures. Generalized quorum systems relax the
connectivity constraints of classical quorum systems: instead of requiring
bidirectional reachability for every pair of write and read quorums, they only
require some write quorum to be \emph{unidirectionally} reachable from some
read quorum. This weak connectivity makes implementing registers particularly
challenging, because it precludes the traditional request/response pattern of
quorum access, making classical solutions like ABD inapplicable. To address
this, we introduce novel logical clocks that allow write and read quorums to
reliably track state updates without relying on bidirectional connectivity.
\end{abstract}

\maketitle

\section{Introduction}
\label{sec:intro}

Tolerating communication channel failures is one of the toughest challenges
facing the developers of modern distributed systems. In fact, according to a
recent study~\cite{osdi-partitions}, a majority of failures attributed to
network partitions led to catastrophic effects whose resolution often required
redesigning core system mechanisms. What makes channel failures even harder to
deal with is that they need to be tolerated in conjunction with ordinary process
failures. The resulting vast space of faulty behaviors makes the analysis of
computability questions under these failure conditions particularly challenging.

It is therefore not surprising that, with a few exceptions, prior work has
studied process and channel failures in isolation from each other. In
particular, it is well-known that tolerating up to $k$ process crashes in a
fully connected network of $n$ processes is possible for a range of problems
(e.g., registers and consensus) if and only if $n \ge 2k+1$~\cite{dls,abd}. Subsequent
work~\cite{flexible-paxos,rambo} generalized this result to the case
when the set of possible faulty behaviors is specified as a \emph{fail-prone
  system} -- a collection of \emph{failure patterns}, each defining a set of
processes that can crash in a single execution~\cite{bqs}. In this case, the
registers and consensus are implementable under a given fail-prone system if there
exists a \emph{read-write quorum system (QS)} in which: any read and write
quorums intersect ({\sf Consistency}); and some read and write quorums of
correct processes are available in every execution ({\sf Availability}). In this
paper we extend these results to failure patterns that may include arbitrary
combinations of both process and channel failures -- namely, process crashes and
channel disconnections.

\begin{example}
\label{ex:channel-failures}
\em Consider a set of processes $\PP=\{a,b,c,d\}$. In Figure~\ref{fig:sample-fs}
we depict a fail-prone system $\FS$, consisting of failure patterns $f_i$,
$i=1..4$ (ignore the sets
$R_i$ and $W_i$ for now). Under failure pattern $f_1$, processes $a$, $b$, $c$
are correct, while $d$ may crash. Channels $(c,a)$, $(a,b)$, $(b,a)$ are
correct, while all others may disconnect.
\end{example}

A plausible conjecture for a tight bound on connectivity under this kind of a
fail-prone system would require the existence of a read-write quorum system
QS$^+$ that preserves {\sf Consistency} but modifies {\sf Availability} to
require that the processes within the available read or write quorum are
strongly connected by correct channels. This ensures that some process can
communicate with both a read and a write quorum (e.g., one in their
intersection), directly enabling the execution of algorithms like ABD~\cite{abd}
and Paxos~\cite{paxos}.  In fact, QS$^+$ was shown to be sufficient for
consensus~\cite{friedman1,friedman2,friemdan-podc-ba,
  aguilera-heartbeat}, and more recently, for registers and consensus under a
more aggressive message loss model~\cite{opodis}.  The latter work also proved
the existence of QS$^+$ to be necessary in the special case of $n = 2k + 1$.

These results, however, leave a noticeable gap. While the existence of QS$^+$ is
optimal for $n=2k+1$, it is unknown whether it is necessary for arbitrary
fail-prone systems, such as those with $k < \lfloor\frac{n-1}{2}\rfloor$, or
those not based on failure thresholds at
all~\cite{quorum-systems-naor,bqs}. Surprisingly, we answer this question in the
negative. We show that atomic registers, atomic snapshots, (single-shot) lattice 
agreement and (partially synchronous) consensus can be implemented even when
none of the available read quorums is strongly connected by correct
channels. Instead, we only require enough 
connectivity for some strongly connected write quorum to be
\emph{unidirectionally} reachable from some read quorum -- a condition that we
formalize via a novel \emph{generalized quorum system (GQS, \S\ref{sec:gqs})}.
We prove that, given a fail-prone system $\FS$ comprising an arbitrary set of
process-channel failure patterns, the above problems are implementable under
$\FS$ if and only if $\FS$ admits a GQS.

\begin{example}
\em
\label{ex:sample-gqs}
Consider the fail-prone system $\FS = \{f_i \mid i=1..4\}$ in
Figure~\ref{fig:sample-fs}. The families of read quorums
$\RR=\{R_i \mid i=1..4\}$ and write quorums $\WW=\{W_i \mid i=1..4\}$ form a
generalized quorum system: each write quorum
$W_i$ is strongly connected and is reachable from
the read quorum $R_i$ through channels correct under the failure
pattern $f_i$. None of the read quorums $R_i$ is strongly connected, thus
relaxing the connectivity requirements of QS$^+$.
\end{example}

Our results also show that any solution to the above-mentioned problems can
guarantee termination only within the write quorums of some GQS (e.g., processes
$a$ and $b$ under the failure pattern $f_1$ from
Figure~\ref{fig:sample-fs}). Such a restricted termination guarantee is expected
in our setting, because channel failures may isolate some correct
processes, making it impossible to ensure termination at all of
them. Accordingly, to prove our upper bounds, we first present an implementation
of atomic registers on top of a GQS that ensures wait-freedom within its write
quorums (\S\ref{sec:upper-bound}). Since atomic snapshots can be constructed
from atomic registers~\cite{hagit-snapshots}, and lattice agreement from
snapshots~\cite{hagit-lattice-agreement}, the upper bounds for snapshots and
lattice agreement follow. 
To prove the lower bounds, we go in the reverse direction: we first prove the
bound for lattice agreement (\S\ref{sec:lower-bound}); then the above-mentioned
constructions imply the bounds for snapshots and registers.

Implementing registers on top of a GQS presents a unique challenge. To complete
a register operation invoked at a process, this process needs to communicate
with both a read and a write quorum -- a typical pattern in algorithms like
ABD. However, the limited connectivity within read quorums means that the
process cannot query a read quorum by simply sending messages to its members and
awaiting responses.

\begin{example}
  \em Assume that a register operation is invoked at process $a$ under failure
  pattern $f_1$ from Figure~\ref{fig:sample-fs}. In this case all channels
  coming into process $c \in R_1$ may have disconnected. This makes it
  impossible for $a$ to request information from this member of the read quorum
  by sending a message to it. Of course, $c$ could periodically push information
  to $a$ through the correct channel $(c,a)$, without waiting for explicit
  requests. But because the network is asynchronous, it is challenging for $a$
  to determine whether the information it receives from $c$ is up to date, i.e.,
  captures all updates preceding the current operation invocation at $a$ -- a
  critical requirement for ensuring linearizability.
\end{example}

We address this challenge using novel logical clocks that processes use to tag
the information they push downstream. These clocks are cooperatively maintained by
processes when updating a write quorum or querying a read quorum. We encapsulate
the corresponding protocol into reusable {\em quorum access functions}, which
are then used to construct an ABD-like algorithm for registers.

Finally, we also show that the existence of a GQS is a tight bound on the
connectivity required for implementing consensus under partial synchrony
(\S\ref{sec:consensus}). Interestingly, implementing consensus
under arbitrary process and channel failures is simpler than implementing
registers, because a process can exploit the eventual timeliness of the network
to determine if the information it receives is up to date.


\section{System Model and Preliminary Definitions}
\label{sec:model}

\begin{figure}[t]
  \centering
  \includegraphics[width=\linewidth,trim=450 150 450 150,clip]{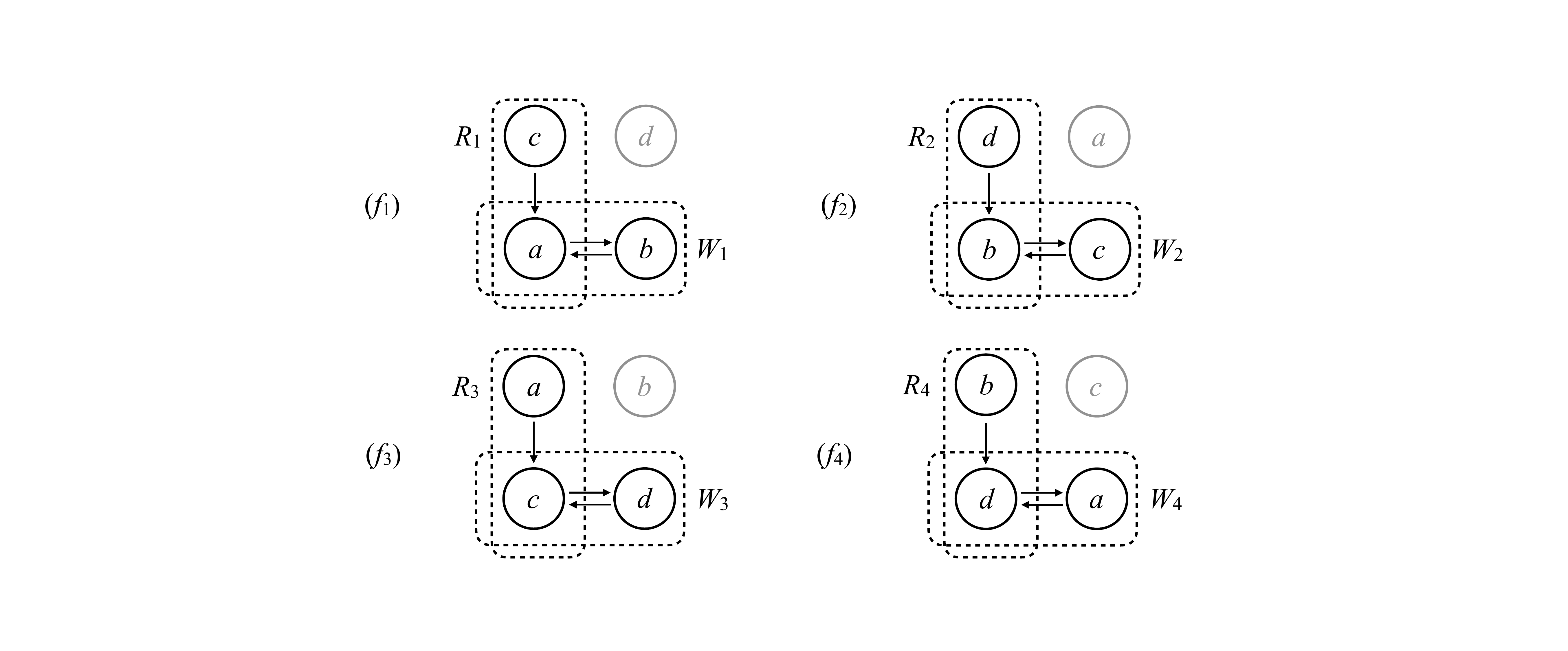}
  \caption{A fail-prone system $\FS=\{f_i \mid i=1..4\}$ and a generalized
    quorum system $(\FS, \RR, \WW)$, for $\RR=\{R_i \mid i = 1..4\}$ and
    $\WW=\{W_i \mid i = 1..4\}$. Solid circles denote correct processes; gray
    circles, processes that may crash; solid arrows, reliable
    channels; missing arrows, channels that may disconnect.}
  \Description{A fail-prone system $\FS=\{f_i \mid i=1..4\}$ and a generalized
    quorum system $(\FS, \RR, \WW)$, for $\RR=\{R_i \mid i = 1..4\}$ and
    $\WW=\{W_i \mid i=1..4\}$. Solid circles denote correct processes; gray
    circles, processes that may crash; solid arrows, reliable
    channels; and missing arrows, channels that may disconnect.}
  \label{fig:sample-fs}
\end{figure}

\paragraph{System model.}
We consider an {\em asynchronous} system with a set $\PP$ of $n$ processes 
that may fail by crashing.
A process is {\em correct} if it never crashes, and {\em faulty} otherwise. 
Processes communicate by exchanging messages through a set of unidirectional 
channels $\CC$:
for every pair of processes $p, q \in \PP$ there is a channel $(p, q) \in \CC$ 
that allows $p$ to send messages to $q$.
Channels can be {\em correct} or {\em faulty}. 
A {\em correct} channel is reliable: it delivers all messages sent by a correct process.
A {\em faulty} channel fails by disconnection: from some point on it drops all
messages sent through it.

\paragraph{Fail-prone systems.}  To state our results, we need a way of
specifying which processes and channels can fail during an execution.
Following our previous work~\cite{opodis}, we do this by generalizing the classical notion of a
fail-prone system~\cite{bqs} to include channel failures. A {\em failure
pattern} is a pair $(P,C) \in 2^\PP \times 2^\CC$ that defines which processes
and channels are allowed to fail in a single execution.
We assume that $C$ contains only
channels between correct processes, as channels incident to faulty processes are
considered faulty by default: $(p,q) \in C \implies \{p,q\} \cap P =
\emptyset$. Given a failure pattern $f = (P, C)$, an execution $\sigma$ of the
system is $f$-\emph{compliant} if at most the processes in $P$ and at most the channels
in $C$ fail in $\sigma$.  A {\em fail-prone system} $\FS$ is a set of failure
patterns (see Figure~\ref{fig:sample-fs}).

\begin{example}
\label{example:std-fail-prone}
\em
The standard failure model where any minority of 
processes may fail and channels between correct processes are reliable
is captured by the fail-prone system
$\FS_M=\{(Q,\emptyset) \mid Q \subseteq \mathcal{P} \wedge |Q| \le \lfloor
\frac{n-1}{2} \rfloor\}$. 
\end{example}

\paragraph{Safety.} We consider algorithms that implement the following
objects in the model just introduced: multi-writer multi-reader (MWMR)
atomic registers, single-writer multi-reader (SWMR) atomic snapshots,
(single-shot) lattice agreement~\cite{hagit-snapshots} and consensus.
These objects provide {\em operations} that can be
invoked by clients -- e.g., reads and writes in the case of registers. Their safety
properties are standard. We introduce them for some of the objects as needed and
defer the rest to \tr{\ref{sec:app-safety-defs}}{\safetydefs}.

\paragraph{Liveness.} Defining liveness properties under our failure model is
subtle, because channel failures may isolate some correct processes, making it
impossible to ensure termination at all of them. To deal with this, we adapt the
classical notions of obstruction-freedom and wait-freedom by parameterizing them
based on the failures allowed and the subsets of correct processes where
termination is required. We use the weaker obstruction-freedom in our lower
bounds and the stronger wait-freedom in our upper bounds.

For a failure pattern $f=(P,C)$ and a set of processes
$T\subseteq \PP\setminus P$, we say that an algorithm $\ALG$ is
$(f,T)$-\emph{wait-free} if, for every process $p\in T$, operation
$\mathit{op}$, and $f$-compliant fair\footnote{Recall that an execution $\sigma$
is \emph{fair} if every process that is correct in $\sigma$ takes an infinite
number of steps in $\sigma$.} 
execution $\sigma$ of $\ALG$, if
$\mathit{op}$ is invoked by $p$ in $\sigma$, then $\mathit{op}$ eventually
returns.
For example, we may require an algorithm to be $(f_1, T_1)$-wait-free for the
failure pattern $f_1$ in Figure~\ref{fig:sample-fs} and $T_1 = \{a, b\}$.
This means that operations invoked at $a$ and $b$ must return despite 
the failures of process $d$ and channels $(a,c)$, $(b,c)$ and $(c,b)$.

An operation $\mathit{op}$ \emph{eventually executes solo} in an execution
$\sigma$ if there exists a suffix $\sigma'$ of $\sigma$ such that all operations
concurrent with $\mathit{op}$ in $\sigma'$ are invoked by faulty processes.
We say that an algorithm $\ALG$ is $(f,T)$-\emph{obstruction-free} if, for every
process $p \in T$, operation $\mathit{op}$, and $f$-compliant fair execution
$\sigma$ of $\ALG$, if $\mathit{op}$ is invoked by $p$ and eventually executes
solo in $\sigma$, then $\mathit{op}$ eventually returns. This notion of
obstruction-freedom aligns with its well-known shared memory
counterparts~\cite{solo-orig,obf-orig,obf-formal}.

We lift the notions of obstruction-freedom and wait-freedom to a fail-prone
system $\FS$ and a {\em termination mapping} $\tau : \FS \rightarrow 2^\PP$ -- a
function mapping each failure pattern $f\in\FS$ to a subset of correct
processes whose operations are required to terminate. We say that an algorithm
$\ALG$ is $(\FS,\tau)$-{\em wait-free} if, for every $f\in\FS$, $\ALG$ satisfies
$(f,\tau(f))$-wait-freedom. We define $(\FS,\tau)$-obstruction-freedom
similarly.

\begin{example}
\em
The standard guarantee of wait-freedom under a minority of process
failures corresponds to $(\FS_M, \tau_M)$-wait-freedom, where $\FS_M$ is
defined in Example~\ref{example:std-fail-prone} and $\tau_M$ selects the set of
all correct processes:
$\forall f=(Q,\emptyset)\in \FS_M.\, \tau_M(f) = \mathcal{P} \setminus Q$.
\end{example}


\section{Generalized Quorum Systems}
\label{sec:gqs}

Fault-tolerant distributed algorithms commonly ensure the consistency of
replicated state using {\em quorums}, i.e., intersecting sets of processes.  It
is common to separate quorums into two classes -- read and write quorums
-- so that the intersection is required only between a pair of quorums from
different classes. For example, in a variant of the ABD register implementation
values are stored at a write quorum and fetched from a read
quorum~\cite{abd,rambo}. The intersection between read and write quorums ensures
that a read operation observes the latest completed write. Similarly, in the
Paxos consensus algorithm decision proposals are stored at a phase-2 quorum
(analogous to a write quorum) and information about previously accepted
proposals is gathered from a phase-1 quorum (analogous to a read
quorum)~\cite{paxos,flexible-paxos}.  The intersection between phase-1 and
phase-2 quorums ensures the uniqueness of decisions.
The classical definition of a quorum system considers only process failures, not
channel failures~\cite{quorum-systems-naor,bqs}. In our framework, we can
express this definition as follows.
\begin{definition}\label{def:classical-qs}
  Consider a fail-prone system $\FS$ that disallows channel failures between
  correct processes: $\forall (P,C)\in\FS.\ C=\emptyset$. A \textbf{quorum
  system} is a triple $(\FS,\RR,\WW)$, where $\RR\subseteq 2^\PP$ is a family
  of read quorums, $\WW\subseteq 2^\PP$ is a family of write quorums, and the
  following conditions hold:
  \begin{itemize}
    \item {\sf Consistency.} For all $R\in\RR$ and $W\in\WW$, $R\cap
      W\neq\emptyset$.
    \item {\sf Availability.} For all $f\in\FS$, there exist $R\in\RR$ and $W\in\WW$
      such that all processes in $R\cup W$ are correct according to $f$.
  \end{itemize}
\end{definition}

\begin{example} 
\label{example:classical-qs}
\em
Consider a system with $n$ processes, where at most
$k \le \lfloor\frac{n-1}{2}\rfloor$ processes can fail. The following triple
$(\FS,\RR,\WW)$ is a quorum system~\cite{flexible-paxos}:
\begin{itemize}
  \item The fail-prone system is such that at most $k$ processes can fail, and channels 
  between correct processes do not fail:\\
  $\FS=\{(P,\emptyset) \mid P \subseteq \PP \wedge |P| \le k\}$.

\item Read quorums are of size at least $n-k$:\\
  $\RR=\{R \mid R \subseteq \PP \wedge |R| \ge n-k\}$.

\item Write quorums are of size at least $k+1$:\\
  $\WW=\{W \mid W \subseteq \PP \wedge |W| \ge k+1\}$.
\end{itemize}
\end{example}

The above quorum system illustrates the usefulness of distinguishing between
read and write quorums: this allows trading off smaller write quorums for larger
read quorums. In the special case where $k = \lfloor\frac{n-1}{2}\rfloor$ we get
$\RR = \WW$, so that both read and write quorums are majorities of processes.

The existence of a classical quorum system is sufficient to implement atomic
registers, atomic snapshots, lattice agreement and consensus in a model without channel
failures.  We now generalize the notion of a quorum system to accommodate such
failures. A straightforward generalization would preserve {\sf Consistency}, but
modify {\sf Availability} to require that the processes within the available
read or write quorum are strongly connected by correct channels. This ensures
that some process can communicate with both a read and a write quorum (e.g., one
in their intersection), directly enabling the execution of algorithms like
ABD. Surprisingly, we find that this strong connectivity requirement is
unnecessarily restrictive: the above-listed problems can be solved even when
read quorums are not strongly connected. To define the corresponding notion of a
quorum system, we use the following concepts:
\begin{itemize}
\item {\em Network graph:} let $\GG=(\PP,\CC)$ be the directed graph with
all processes as vertices and all channels as edges.
\item {\em Residual graph:} for a failure pattern $f=(P,C)$, let
  $\GG\setminus f$ be the subgraph of $\GG$ obtained by removing all processes
  in $P$, their incident channels, and all channels in $C$.
\item {\em $f$-availability:} a set $Q\subseteq \PP$ is $f$-available if it
  contains only processes correct according to $f$ and it is strongly connected in $\GG\setminus f$.
  This implies that all processes in $Q$ can communicate with each other via
  channels correct according to $f$.
\item {\em $f$-reachability:} a set $W\subseteq \PP$ is $f$-reachable from a set $R\subseteq \PP$ if both
  $W$ and $R$ contain only processes correct according to $f$, and every member of $W$ can be
  reached by every member of $R$ via a directed path in $\GG\setminus f$.
\end{itemize}

\begin{example}
\label{example:fail-prone-system}
\em In Figure~\ref{fig:sample-fs}, for each $i=1..4$, $W_i$ is $f_i$-available,
and $W_i$ is $f_i$-reachable from $R_i$.
\end{example}

\begin{definition}
  \label{def:gqs}
  A \textbf{generalized quorum system} is a triple $(\FS,\RR,\WW)$, where
  $\FS$ is a fail-prone system, $\RR\subseteq 2^\PP$ is a family of read quorums,
  $\WW \subseteq 2^\PP$ is a family of write quorums, and the following conditions 
  hold:
\begin{itemize}
\item {\sf Consistency.} For every $R\in\RR$ and $W\in\WW$, 
  $R\cap W\neq\emptyset$.
\item {\sf Availability.} For all $f\in\FS$, there exist $W\in\WW$ and $R\in\RR$
  such that $W$ is $f$-available, and $W$ is $f$-reachable from $R$.
\end{itemize}
\end{definition}
Informally speaking, {\sf Availability} guarantees that an operational write
quorum can unidirectionally receive information from an operational read quorum
under any failure scenario defined by $\FS$.

A classical quorum system is a
special case of a generalized quorum system.  Indeed, if $\FS$ disallows channel
failures between correct processes, then every correct write quorum can be
trivially reached from every read quorum, and Definition~\ref{def:gqs} becomes
equivalent to Definition~\ref{def:classical-qs}.

\begin{example}
\label{example:gqs}
\em
  Consider the fail-prone system $\FS = \{f_i \mid i=1..4\}$ in
  Figure~\ref{fig:sample-fs} and let $\WW=\{W_i \mid i=1..4\}$ and
  $\RR=\{R_i \mid i=1..4\}$. Then the triple $(\FS,\RR,\WW)$ is a generalized
  quorum system.
  Indeed:
\begin{itemize}
  \item {\sf Consistency}. For each $i, j = 1..4$, $R_i\cap W_j\neq\emptyset$.
  \item {\sf Availability}. For each $i=1..4$, $W_i$ is $f_i$-available, and
    $W_i$ is $f_i$-reachable from $R_i$.
  \end{itemize}
\end{example}

Note that the above $(\FS,\RR,\WW)$ is a valid generalized quorum system even
though the processes in each read quorum are not strongly connected via correct
channels: some pair of processes are only connected unidirectionally. 
This relaxation allows read quorums to be formed under a broader range of failure
scenarios. In contrast, processes within each write quorum validating {\sf
  Availability} must be strongly connected via correct channels. In fact, for a
given failure pattern $f$ we can show that different write quorums validating
{\sf Availability} with respect to $f$ must also be strongly connected via
correct channels.

\begin{proposition}
  \label{proposition:f-avail-scc}
  Let $(\FS,\RR,\WW)$ be a generalized quorum system.  For each $f\in\FS$,
  the following set of processes is strongly connected in $\GG\setminus f$:
\begin{align*}
U = \bigcup\{W \mid {} & {W \in \WW} \wedge (W \text{ is } f\text{-available}) \land {}\\
& \exists R\in\RR.\, (W \text{ is } f\text{-reachable from } R)\}.
\end{align*}
\end{proposition}
\begin{proof}
  Fix $f\in\FS$ and let $p_1,p_2\in U$. 
  Then there exist $W_1,W_2\in\WW$ such that
  \begin{itemize}
    \item $(p_1\in W_1) \wedge (W_1 \text{ is } f\text{-available})\land {}$\\
      $\exists R_1\in\RR.\, (W_1 \text{ is } f\text{-reachable from } R_1)$; and
    \item $(p_2\in W_2) \wedge (W_2 \text{ is } f\text{-available})\land {}$\\
      $\exists R_2\in\RR.\, (W_2 \text{ is } f\text{-reachable from } R_2)$.
    \end{itemize}
  Because any read and write quorums intersect, $W_1\cap R_2\neq\emptyset$.
  Fix $q \in W_1\cap R_2$.
  Since $W_1$ is $f$-available, there exists a path via correct channels from $p_1$ to $q$.
  And since $W_2$ is $f$-reachable from $R_2$, there exists a path via correct channels from $q$
  to $p_2$.
  Thus, $p_2$ is reachable from $p_1$ via correct channels.
  We can analogously show that $p_1$ is reachable from $p_2$ via correct channels, as required.
\end{proof}

For $f\in\FS$ we denote the strongly connected component of $\GG\setminus f$
containing $U$ by $\ter{f}$: this component includes all write quorums that
validate {\sf Availability} with respect to $f$. The {\sf Availability} property
of the generalized quorum system ensures that $\ter{f} \not= \emptyset$.


\section{Main Results}
\label{sec:gqs:main}

We now state our main results for asynchrony, proved in
\S\ref{sec:upper-bound}-\ref{sec:lower-bound}: the existence of a generalized
quorum system is a tight bound on the process and channel failures that can be
tolerated by any implementation of MWMR atomic registers, SWMR atomic snapshots,
and lattice agreement (we handle partially synchronous consensus in
\S\ref{sec:consensus}). The following theorem establishes that the existence of
a generalized quorum system is sufficient to implement each of the three
objects.  For each $f\in\FS$, these implementations provide wait-freedom within
the strongly connected component $\ter{f}$.

\begin{theorem}
\label{thm:reg-upper}
Let $(\FS,\RR,\WW)$ be a generalized quorum system and $\tau : \FS \rightarrow 2^{\PP}$ 
be the termination mapping such that for each $f\in\FS$, $\tau(f)=\ter{f}$.
Then there exists an $(\FS,\tau)$-wait-free implementation for each of the
following objects: MWMR atomic registers, SWMR atomic snapshots, and lattice agreement.
\end{theorem}

The next theorem establishes a matching lower bound, 
showing that the existence of a generalized quorum system is necessary to implement 
any of the three objects. 
This assumes a weak termination guarantee that only requires 
obstruction-freedom to hold at some non-empty set of processes for each failure pattern.

\begin{theorem}
\label{thm:reg-lower}
Let $\mathcal{F}$ be a fail-prone system and $\tau : \FS \rightarrow 2^{\PP}$ 
be a termination mapping such that for each $f\in\FS$, $\tau(f)\neq\emptyset$.
Assume that there exists an $(\FS,\tau)$-obstruction-free implementation of any of the
following objects: MWMR atomic registers, SWMR atomic snapshots, or lattice agreement.
Then there exist $\RR$ and $\WW$ such that $(\FS,\RR,\WW)$ is a generalized quorum system.
Moreover, for each $f\in\FS$, we have $\tau(f)\subseteq \ter{f}$.
\end{theorem}

The theorem also shows $\ter{f}$ is the largest set of processes for which
termination can be guaranteed under the failure pattern $f$. Furthermore, the
two theorems imply that if termination can be guaranteed for at least one
process, then it can also be guaranteed for all processes in $\ter{f}$.

\begin{example}
  \em Consider the generalized quorum system $(\FS,\RR,\WW)$ from
  Figure~\ref{fig:sample-fs}. Then $U_{f_1}=\{a,b\}$, $U_{f_2}=\{b,c\}$,
  $U_{f_3}=\{c,d\}$ and $U_{f_4}=\{d,a\}$.  Consider $\tau$ such that
  $\tau(f_i) = U_{f_i}$, $i=1..4$. By Theorem~\ref{thm:reg-upper}, there exists
  an $(\FS,\tau)$-wait-free implementation for any of the three objects
  considered.  Suppose now that we change $f_1$ to $f'_1$ that additionally
  fails the channel $(a, b)$. Let $\FS'=\{f'_1,f_2,f_3,f_4\}$. It is easy to
  check that there do not exist $\RR'$ and $\WW'$ that would form a generalized
  quorum system $(\FS',\RR',\WW')$. Thus, Theorem~\ref{thm:reg-lower} implies
  that there is no implementation of any of the three objects that would provide
  obstruction-freedom (and, by extension, wait-freedom) anywhere under $\FS'$.
\end{example}

To prove the upper bound (Theorem~\ref{thm:reg-upper}), we construct an
$(\FS,\tau)$-wait-free implementation of MWMR atomic registers
(\S\ref{sec:upper-bound}). Since SWMR atomic snapshots can be constructed from
the registers~\cite{hagit-snapshots}, and  lattice agreement can in turn be
constructed from snapshots~\cite{hagit-lattice-agreement}, we naturally get the
upper bounds for the latter two problems. 
We consider the weakest variant of lattice
agreement~\cite{hagit-lattice-agreement}, which is single-shot and therefore
cannot be used for implementing multi-shot objects, such as registers.
Thus, to prove the lower bound (Theorem~\ref{thm:reg-lower}), we first establish
it for lattice agreement (\S\ref{sec:lower-bound}); then the above-mentioned
constructions imply the lower bound for snapshots and registers.


\section{Upper Bound for Atomic Registers}
\label{sec:upper-bound}

Fix a generalized quorum system $(\FS,\RR,\WW)$ and a termination mapping
$\tau : \FS \rightarrow 2^{\PP}$ such that $\tau(f)=\ter{f}$ for each
$f \in \FS$. Without loss of generality, we assume that the connectivity
relation of the graph $\GG\setminus f$ is transitive for each $f \in \FS$: if
not, transitivity can be easily simulated by having all processes forward every
received message. 
To implement atomic registers using the quorum system, we need to
come up with a way for a process to contact a read and a write quorum, as
required by algorithms such as ABD~\cite{abd}. This is challenging in our
setting because processes within a read quorum may not be strongly connected 
by correct channels.

\begin{example}
\label{example:gqs-2}
\em Consider the generalized quorum system in Figure~\ref{fig:sample-fs} and the
failure pattern $f_1$. The available read and write quorums are $R_1=\{a,c\}$
and $W_1=\{a,b\}$. Since Theorem~\ref{thm:reg-upper} requires
$\tau(f_1)=\ter{f_1}$, a register implementation validating it has to ensure
wait-freedom within $W_1$. However, all channels coming into $c$ may have
failed. This makes it impossible for a member of $W_1$, such as $a$, to request
information from some of the members of $R_1$, such as $c$, by sending them an
explicit message to this end.
\end{example}

\paragraph{Quorum access functions.} We encapsulate the mechanics necessary
to deal with this challenge using the following \emph{quorum access functions},
which allow a process to obtain up-to-date information from a quorum. We assume
that the top-level protocol, such as a register implementation, maintains a
state from a set $\ST$ at each process. Then the interface of the access functions
is as follows:
\begin{itemize}
\item $\readquorum() \in 2^\ST$: returns the states of all members of some read
  quorum; and
\item $\writequorum(\upd)$: applies the {\em update function}
  $\upd : \ST \to \ST$ to the states of all members of some write quorum.
\end{itemize}

We require that these functions satisfy the following properties:
\begin{itemize}
\item {\sf Validity.}  For any state $s$ returned by $\readquorum()$,
  there exists a subset of previous invocations
  $\{\writequorum(u_i) \mid i=1..k\}$ such that $s$ is the result of applying
  the update functions in $\{u_i \mid i=1..k\}$ to the initial state in some
  order.
\item {\sf Real-time ordering.} If $\writequorum(\upd)$
  terminates, then its effect is visible to any later $\readquorum()$,
  i.e., the set returned by $\readquorum()$ includes at least one state to
  which $\upd$ has been applied.
\item {\sf Liveness.} The functions are $(\FS,\tau)$-wait-free, i.e.,
  they terminate at every member of $\tau(f)$ for any $f \in \FS$.
\end{itemize}

As we show in the following, quorum access functions are sufficient to program
an ABD-like algorithm for atomic registers. 

\paragraph{Quorum access functions for classical quorum systems.}
To illustrate the concept of quorum
access functions, in Figure~\ref{code:qaf-simpler} we provide their
implementation using a classical quorum system that disallows channel failures
(Definition~\ref{def:classical-qs}). 
Each process stores the state of the
top-level protocol, such as a register implementation, in a $\State$
variable. This state is managed by the implementation of the quorum access
functions, but its structure is opaque to this implementation: it
can only manipulate the state by applying update functions passed by the
callers. Each process also maintains a monotonically increasing $\seq$ number
(initially $0$),
which is used to generate a unique identifier for each quorum access function
invocation at this process. This identifier is then added to every message
exchanged by the implementation, so that the process can tell which messages
correspond to which invocations.

Upon a call to $\readquorum()$ at a process, it broadcasts a $\GETREQ$ message
(line~\ref{code:qaf-simpler:readquorum:getreq}). Any process receiving this
responds with a $\GETRESP$ message that carries its current state
(line~\ref{code:qaf-simpler:get-response}). The $\readquorum()$ invocation
returns when it accumulates such responses from a read quorum
(line~\ref{code:qaf-simpler:readquorum:receive-getresp}). Upon a call to
$\writequorum(\upd)$ at a process, it broadcasts a $\SETREQ$ message carrying
the function $\upd$ (line~\ref{code:qaf-simpler:setreq}). Any process receiving
this applies $\upd$ to its current state and responds with $\SETRESP$
(line~\ref{code:qaf-simpler:send-setresp}). The $\writequorum()$ invocation
returns when it accumulates such responses from a write quorum
(line~\ref{code:qaf-simpler:writequorum:receive-getresp}).

\begin{figure}[t]
	\removelatexerror
	\begin{algorithm*}[H]
		\DontPrintSemicolon
		\SetAlgoNoLine
		\setcounter{AlgoLine}{0}

		$\State\in\mathcal{S}$\quad \emph{// opaque state of the top-level protocol}\;
		$\assign{\seq}{0}$\;
		
    \smallskip
    \smallskip

		\Function{$\readquorum()$}{\label{code:qaf-simpler:readquorum}
			$\assign{\seq}{\seq + 1}$\;
			\send $\GETREQ(\seq)$ \ToAll\;\label{code:qaf-simpler:readquorum:getreq}
			\waituntil $\{\GETRESP(\seq,s_j) \mid p_j\in R\}$
                        \label{code:qaf-simpler:readquorum:receive-getresp}\linebreak 
				\phantom\ \ \ \bf{from some} $R\in\RR$\;
			\return $\{s_j \mid p_j\in R \}$\;
		}

		\smallskip
		
		\SubAlgo{\onreceive $\GETREQ(k)$ \from $p_j$\label{code:qaf-simpler:get-version}}{
	    	\send $\GETRESP(k,\State)$ \To $p_j$\label{code:qaf-simpler:get-response}\;
	    }

		\smallskip

		\Function{$\writequorum(\upd)$}{\label{code:qaf-simpler:writequorum}
			$\assign{\seq}{\seq + 1}$\;
			\send $\SETREQ(\seq,\upd)$ \ToAll\; \label{code:qaf-simpler:setreq}
			\waituntil $\{\SETRESP(\seq) \mid p_j\in W\}$
                        \label{code:qaf-simpler:writequorum:receive-getresp}\linebreak
				\phantom\ \ \ \bf{from some} $W\in\WW$\;
		}

		\smallskip

		\SubAlgo{\onreceive $\SETREQ(k,\upd)$ \from $p_j$\label{code:qaf-simpler:set-version}}{
			$\assign{\State}{\upd(\State)}$\;
			\send $\SETRESP(k)$ \To $p_j$\; \label{code:qaf-simpler:send-setresp}
		}

  \end{algorithm*}
	\caption{Quorum access functions for a classical quorum system: the protocol
          at a process $p_i$.}
  \Description{Quorum access functions for a classical quorum system: the protocol
          at a process $p_i$.}
	\label{code:qaf-simpler}
\end{figure}

It is easy to see that the above implementation ensures the {\sf
Validity} and {\sf Real-time ordering} properties, the latter because
any read and write quorums intersect. Finally, this implementation guarantees
wait-freedom at every correct process ({\sf Liveness}): {\sf Availability}
ensures that there exist a read quorum and a write quorum of correct processes;
then the fact that the fail-prone system disallows channel failures ensures that
any process can communicate with these quorums.

\paragraph{Quorum access functions for generalized quorum systems.}  We now
show how we can implement the quorum access functions for generalized quorum
systems, despite the lack of strong connectivity within read quorums. We use
Example~\ref{example:gqs-2} for illustration. Recall that in this example we
need to ensure termination within $W_1$ and, in particular, at $a$. An
implementation of $\readquorum()$ at $a$ needs to obtain state
snapshots from every member of $R_1$. However, channel failures may make it
impossible for $a$ to send a message to $c \in R_1 \setminus W_1$ to request
this information. Of course, $c$ could just periodically propagate its state to
all processes it is connected to, without them asking for it explicitly: by {\sf
Availability}, $W_1$ is $f$-reachable from $R_1$, so that messages sent by $c$
will eventually reach $a$. But because the network is asynchronous,
\emph{the process $a$ cannot easily determine when the information it receives is up
to date}, i.e., when it captures the effects of all $\writequorum()$
invocations that completed before $\readquorum()$ was called at $a$ -- 
as necessary to satisfy {\sf Real-time ordering}.

To address this challenge, each process maintains a monotonically increasing
logical clock, stored in the variable $\clock$
(initially 0; this clock is different from the usual logical clocks for tracking
causality~\cite{lamport-clocks,vectorclocks1}). We then modify the
implementation of the quorum access functions in Figure~\ref{code:qaf-simpler}
as shown in Figure~\ref{code:qaf-complex}:

\begin{itemize}
\item {\bf Periodic state propagation (line~\ref{code:qaf-complex:periodically}):} 
	A process periodically advances its $\clock$ and propagates its current
  $\State$ and $\clock$ in a $\GETRESP$ message, without waiting for an explicit
  request. The clock value in the
  message indicates the logical time by which the process had this state.

\item {\bf Clock updates during state changes (line~\ref{reg:set-version}):} 
	When handling a $\SETREQ(k, \upd)$ message, a process increments its
  $\clock$ and sends it as part of the $\SETRESP$ message.
  The process thereby indicates the logical time by which it has incorporated $\upd$ into its state.

\item {\bf Delaying the completion of $\writequorum$
    (line~\ref{code:qaf-complex:write}):}  
	Upon a $\writequorum(\upd)$ invocation,
  the process broadcasts $\upd$ in a
  $\SETREQ$ message and waits for $\SETRESP$ messages from every member of some
  write quorum $\Wset$, as in the classical implementation. The process selects
  the highest clock value among the responses received and stores it in a variable
  $\cset$.
  It then waits until some read quorum $\Rset$ reports having
  $\clock\geq\cset$ before completing the invocation. As we show in the
  following, this wait serves to ensure that future $\readquorum()$ invocations
  will observe the update $\upd$.
  
\item {\bf Clock cutoff for $\readquorum$ (line~\ref{code:qaf-complex:read}):} 
	Upon a $\readquorum()$ invocation, the
  process first determines a clock value $\cget$, delimiting how up-to-date the
  states it will return should be. To this end, the process broadcasts a
  $\CLOCKREQ$ message. Any process receiving this message responds with a $\CLOCKRESP$
  message that carries its current $\clock$ value. The process executing
  $\readquorum()$ waits until it receives such responses from all members of
  some write quorum and picks the highest clock value among those
  received -- this is the desired clock cut-off $\cget$. Finally, the process
  waits until it receives $\GETRESP(s_j,c_j)$ messages with $c_j\geq \cget$ from
  all members of some read quorum and returns the collected
  states $s_j$ to the caller.
\end{itemize}

\begin{figure}[t]
	\removelatexerror%
	\begin{minipage}{1.1\linewidth}
	\begin{algorithm*}[H]
		\DontPrintSemicolon
		\SetAlgoNoLine
		\setcounter{AlgoLine}{0}

    $\State\in\mathcal{S}$\quad \emph{// opaque state of the top-level protocol}\;
		$\assign{\seq,\clock}{0,0}$\;
		
    \smallskip
    \smallskip
                
		\Function{$\readquorum()$\label{code:qaf-complex:read}}{
			$\assign{\seq}{\seq + 1}$\;
			\send $\CLOCKREQ(\seq)$ \ToAll\;
			\waituntil $\{\CLOCKRESP(\seq,c_j) \mid p_j\in \Wget\}$\linebreak
				\phantom\ \ \ \bf{from some} $\Wget\in\WW$\;\label{code:qaf-complex:receive-clock}
			$\assign{\cget}{\max\{c_j \mid p_j\in \Wget\}}$\;\label{code:qaf-complex:read-quorum:cmin}
			\waituntil $\{\GETRESP(s_j,c_j) \mid p_j\in \Rget\}$\linebreak
				\phantom\ \ \ \bf{from some} $\Rget\in\RR$ \bf{where} $\forall j.\  c_j \geq \cget$
        \;\label{code:qaf-complex:read-quorum:wait-read-quorum}
			\return $\{s_j \mid p_j\in \Rget\}$\;
		}		

		\smallskip

		\SubAlgo{\onreceive $\CLOCKREQ(k)$
	    	\from $p_j$}{\label{reg:get-version}
	    	\send $\CLOCKRESP(k,\clock)$ \To $p_j$\; \label{code:qaf-complex:send-clockresp}
	    }

		\smallskip

		\SubAlgo{{\bf periodically}}{\label{code:qaf-complex:periodically}
			$\assign{\clock}{\clock+1}$\;
			\send $\GETRESP(\State,\clock)$ \ToAll\; \label{code:qaf-complex:send-getresp}
    }
	
		\smallskip

		\Function{$\writequorum(\upd)$\label{code:qaf-complex:write}}{
			$\assign{\seq}{\seq + 1}$\; \label{code:qaf-complex:write-seq-inc}
			\send $\SETREQ(\seq,\upd)$ \ToAll\;
			\waituntil $\{\SETRESP(\seq,c_j) \mid p_j\in \Wset\}$ \linebreak
				\phantom\ \ \ \bf{from some} $\Wset\in\WW$\;\label{code:qaf-complex:write-quorum:wait}
			$\assign{\cset}{\max\{c_j \mid p_j\in \Wset\}}$\;\label{code:qaf-complex:write-quorum:cmin}
			\waituntil $\{\GETRESP(\_,c_j) \mid p_j\in \Rset\}$ \linebreak
				\phantom\ \ \ \bf{from some} $\Rset\in\RR$ \bf{where} $\forall j.\  c_j \geq \cset$
        \;\label{code:qaf-complex:write-quorum:wait-read-quorum}
		}

		\smallskip

		\SubAlgo{\onreceive $\SETREQ(k,\upd)$ \from $p_j$\label{reg:set-version}}{
			$\assign{\State}{\upd(\State)}$\; \label{code:qaf-complex:write-quorum:set-state}
			$\assign{\clock}{\clock+1}$\; \label{code:qaf-complex:write-quorum:inc-clock}
			\send $\SETRESP(k,\clock)$ \To $p_j$\; \label{code:qaf-complex:write-quorum:setresp}
    }
	\end{algorithm*}
  \end{minipage}
	\caption{Quorum access functions for a generalized quorum system:
          the protocol at a process $p_i$.}
  \Description{Quorum access functions for a generalized quorum system:
          the protocol at a process $p_i$.}
	\label{code:qaf-complex}
\end{figure}

Note that $\writequorum$ and $\readquorum$ operations work in tandem:
$\writequorum$ delays its completion until clocks have advanced sufficiently at
a read quorum; this allows $\readquorum$ to establish a clock cutoff capturing
all prior completed updates. Interestingly, $\writequorum$ uses read quorums for
this purpose, while $\readquorum$ uses write quorums -- an inversion of the
traditional quorum roles.

It is easy to see that this implementation validates the {\sf Validity}
property of the quorum access functions. We now prove that it also validates
{\sf Real-time ordering}. First, consider $\cset$ computed by a
$\writequorum(\upd)$ invocation
(lines~\ref{code:qaf-complex:write-seq-inc}-\ref{code:qaf-complex:write-quorum:cmin}).
The following lemma shows that querying the states of a read quorum with clocks
$\ge \cset$ is sufficient to observe $\upd$.
\begin{lemma}
\label{lemma:read-up-to-date}
Assume that $\cset$ is computed by a $\writequorum(\upd)$ invocation at a
process $\pset$
(lines~\ref{code:qaf-complex:write-seq-inc}-\ref{code:qaf-complex:write-quorum:cmin}).
Consider a set of messages $\{\GETRESP(s_j,c_j) \mid p_j \in R\}$ sent by all
members of some read quorum $R$, each with $c_j\geq\cset$. Then some $s_j$ has
incorporated $\upd$.
\end{lemma}

\begin{proof}
  To compute the value of $\cset$, the process $\pset$ waits for
  $\SETRESP(\_,c_j')$ messages from all members $p_j$ of a write quorum $\Wset$.
  Since any read quorum intersects any write quorum, there exists 
  $p_l\in R \cap\Wset$.
  Because $p_l\in\Wset$, this process sends a $\SETRESP(\_,c'_l)$ message to $\pset$ during
  the $\writequorum(u)$ invocation.
  The process $\pset$ computes $\cset=\max\{c_j' \mid p_j\in \Wset\}$,
  so that $\cset\geq c'_l$.
  Because $p_l\in R$, this process also sends a $\GETRESP(s_l,c_l)$ message with
  $c_l\geq\cset \geq c'_l$.
  At this moment $p_l$ has $\clock = c_l$.
  Hence, if $p_l$ sends the $\GETRESP(s_l,c_l)$ message before sending 
  the $\SETRESP(\_,c'_l)$ message, then
  due to the increment at line~\ref{code:qaf-complex:write-quorum:inc-clock}
  and the fact that process clocks never decrease, 
  we must have $c_l<c'_l$. But this contradicts the fact $c_l\geq c'_l$ that we
  established earlier.
  Hence, $p_l$ must send the $\SETRESP(\_,c'_l)$ message before sending
  the $\GETRESP(s_l,c_l)$ message and, therefore, $s_l$ incorporates $u$.
\end{proof}

In the light of the above lemma, for $\readquorum()$ to validate {\sf Real-time
ordering}, it just needs to find a clock value that is $\ge \cset$ of any
previously completed $\writequorum()$. As we show in the following proof, this
is precisely what is achieved by the computation of $\cget$ in $\readquorum()$.

\begin{theorem}[{\sf Real-time ordering}]
\label{thm:qaf-safety}
If a $\writequorum(\upd)$ operation terminates, then its effect is visible to
any subsequent $\readquorum()$ invocation.
\end{theorem}

\begin{proof}
  Assume that $\writequorum(\upd)$ terminates at a process $\pset$ before
  $\readquorum()$ is invoked at a process $\pget$. Since any read quorum
  intersects any write quorum, there exists $p_l\in \Rset\cap\Wget$. Since
  $p_l \in \Rset$, before $\writequorum(\upd)$ terminated, $\pset$ received
  $\GETRESP(\_, c_l')$ from $p_l$ with $c_l' \ge \cset$
  (line~\ref{code:qaf-complex:write-quorum:wait-read-quorum}). Hence, by the
  time $\writequorum(\upd)$ terminated, $p_l$ had $\clock \ge \cset$. We also
  have $p_l \in \Wget$, and thus $\pget$ received $\CLOCKRESP(\_, c_l)$ from
  $p_l$ at line~\ref{code:qaf-complex:receive-clock}. This message was sent at
  line~\ref{code:qaf-complex:send-clockresp} after $\readquorum()$ had been
  invoked, and thus, after $\writequorum(\upd)$ had terminated. Above we
  established that by the latter point $p_l$ had $\clock \ge \cset$, and thus,
  $c_l \geq \cset$. The process $\pget$ computed
  $\cget = \max\{c_j \mid p_j\in \Wget\}$ at
  line~\ref{code:qaf-complex:read-quorum:cmin}, so that
  $\cget\geq c_l\geq\cset$. Then each $\GETRESP(s_j,c_j)$ that $\pget$ received
  from $p_j \in \Rget$ at
  line~\ref{code:qaf-complex:read-quorum:wait-read-quorum} satisfies
  $c_j\geq\cget\geq\cset$. By Lemma~\ref{lemma:read-up-to-date}, some $s_j$ has
  incorporated $\upd$, as required.
 \end{proof}

\begin{theorem}[{\sf Liveness}]
\label{thm:qaf-liveness}
The protocol in Figure \ref{code:qaf-complex} is $(\FS,\tau)$-wait-free.
\end{theorem}

\begin{proof}
  We prove the liveness of $\readquorum()$; the case of $\writequorum()$ is
  analogous. Fix a failure pattern $f\in\FS$ and a process $p\in\tau(f)$ that
  executes $\readquorum()$. 
  By {\sf Availability}, there exist $W\in\WW$ and $R \in \RR$ such that 
  $W$ is $f$-available, and $W$ is $f$-reachable from $R$.
  By Proposition~\ref{proposition:f-avail-scc}, $W \subseteq \ter{f}$.
  Then since $\tau(f)=\ter{f}$ and $p \in \tau(f)$, 
  $p$ is strongly connected to $W$ via channels correct under $f$.
  Therefore, $p$ will
  be able to exchange the $\CLOCKREQ$ and $\CLOCKRESP$ messages with every
  member of $W$, thus exiting the wait at
  line~\ref{code:qaf-complex:receive-clock}. Recall that $W$ is $f$-reachable
  from $R$ and each process periodically increments its $\clock$ value and
  propagates it in a $\GETRESP$ message
  (line~\ref{code:qaf-complex:periodically}). Hence, $p$ will receive $\GETRESP$ 
  messages from all members of $R$ with high enough clock values to exit the
  wait at line~\ref{code:qaf-complex:read-quorum:wait-read-quorum} and return.
\end{proof}

\begin{figure}[t]
	\removelatexerror
	\begin{algorithm*}[H]
		\DontPrintSemicolon
		\SetAlgoNoLine
		\setcounter{AlgoLine}{0}

		$\mathcal{S}=\Val\times\mathsf{Version}$\quad \emph{// register
                  state type}\;

		\smallskip
		\smallskip
              
		\Function{$\owrite(x)$}{\label{code:reg:write}
			$\assign{S}{\readquorum()}$\;\label{code:reg:write:read-quorum}
			$\assign{(k,\_)}{\max\{s.\ts \mid s \in S\}}$\;\label{code:reg:t'}
			$\assign{t}{(k+1,i)}$\;\label{code:reg:t}
			\mbox{$\assign{\upd}{(\lambda s.\ }$\lIf{$t > s.\ts$}{\label{code:reg:write:upd}
					\return $(x,t)$} \lElse{\return $s)$}}\;
			$\writequorum(\upd)$\;
	    }
		
		\smallskip

		\Function{$\oread()$}{\label{code:reg:read}
			$\assign{S}{\readquorum()}$\;\label{code:reg:read:read-quorum}
			{\bf let $s' \in S$ be such that} $\forall s \in S.\ s'.\ts \geq s.\ts$\;\label{code:reg:read:t}
			\mbox{$\assign{\upd}{(\lambda s.\ }$\lIf{$s'.\ts > s.\ts$}{\label{code:reg:read:upd}
					\return $s'$} \lElse{\return $s)$}}\!\!\!\;
			$\writequorum(\upd)$\;\label{code:reg:read:write-quorum}
	        \return $s'.\val$\;
    }
	\end{algorithm*}
	\caption{The atomic register protocol at a process $p_i$.}
  \Description{The atomic register protocol at a process $p_i$.}
	\label{code:reg}
\end{figure}

\paragraph{Registers via quorum access functions.} In Figure~\ref{code:reg}
we give an implementation of an atomic register using the above quorum access
functions, which validates Theorem~\ref{thm:reg-upper}. The implementation
follows the structure of a multi-writer/multi-reader variant of
ABD~\cite{abd,rambo}: the main novelty of our protocol lies in the
implementation of the quorum access functions.

Let $\Val$ be the domain of values the register stores. Like ABD, our
implementation tags values with versions from
$\mathsf{Version}=\mathbb{N}\times\mathbb{N}$, ordered
lexicographically. A version is a pair of a monotonically increasing number and
a process identifier. Each register process maintains a state consisting of a
pair $(\val,\ts)$, where $\val$ is the most recent value written to the register
at this process and $\ts$ is its version.  Hence, we instantiate $\ST$ in
Figure~\ref{code:qaf-complex} to $\ST=\Val\times\mathsf{Version}$,
with $(0, 0)$ as the initial state.  
To execute a $\owrite(x)$ operation (line~\ref{code:reg:write} in
Figure~\ref{code:reg}), a process proceeds in two phases:

\begin{itemize}
\item {\bf Get phase.} The process uses $\readquorum()$ to collect the states
  from some read quorum. Based on these, it computes a unique version $t$,
  higher than every received one.
\item {\bf Set phase.} The process next uses $\writequorum()$ to store the pair
  $(x,t)$ at some write quorum. To this end, it passes as an argument a function
  $\upd$ that describes how each member of the write quorum should update its state
  (expressed using $\lambda$-notation,
  line~\ref{code:reg:write:upd}). Given a state
  $s$ of a write-quorum member, the function acts as follows: if the new
  version $t$ is higher than the old version $s.\ts$, then the function
  returns a state with the new value $x$ and version $t$; otherwise it returns
  the unchanged state $s$. Recall that the implementation of $\writequorum()$
  uses the result of this function to replace the states at the members of a
  write quorum.
\end{itemize}

To execute a $\oread()$ operation (line~\ref{code:reg:read}), a process follows
two similar phases:

\begin{itemize}
\item {\bf Get phase.} The process uses $\readquorum()$ to collect the states
  from all members of some read quorum. It then picks the state $s'$ with the
  largest version among those received. The value part $s'.\val$ of this state
  will be returned as a response to the $\oread()$.
\item {\bf Set phase.} Before returning from $\oread()$, the process must
  guarantee that the value read will be seen by any subsequent operation. To
  this end, it writes the state $s'$ back using similar steps to the {\sf Set
  phase} for the $\owrite()$ operation.
\end{itemize}

{\sf Liveness} of the quorum access functions trivially implies that the
protocol in Figure~\ref{code:reg} is $(\FS,\tau)$-wait-free. The {\sf Real-time
ordering} property can be used to show that the protocol is
linearizable~\cite{linearizability}; we defer the proof to
\tr{\ref{sec:reg-safety}}{\regsafety}.


\section{Lower Bound for Lattice Agreement}
\label{sec:lower-bound}

In this section we prove Theorem~\ref{thm:reg-lower} for lattice agreement. Then
the existing constructions of lattice agreement from snapshots and
registers~\cite{hagit-snapshots,hagit-lattice-agreement} imply that the theorem
holds also for the latter objects. Recall that in the lattice agreement
problem, each process $p_i$ may invoke an operation $\propose(x_i)$ with its
input value $x_i$. This invocation terminates with an output value $y_i$. Both input
and output values are elements of a semi-lattice with a partial order $\leq$ and
a join operation $\bigsqcup$. An algorithm that solves lattice agreement must
satisfy the following conditions for all $i$ and $j$:
\begin{itemize}  
\item \textsf{Comparability.} Either $y_i \leq y_j$ or $y_j \leq y_i$.
\item \textsf{Downward validity.} If process $p_i$ outputs $y_i$, then
  $x_i \leq y_i$.
\item \textsf{Upward validity.} If process $p_i$ outputs $y_i$, then
  $y_i \leq \bigsqcup X$, where $X$ is the set of $x_j$ for which
  $\text{propose}(x_j)$ was invoked.
\end{itemize}

We rely on the following lemma, which establishes that processes where
obstruction-freedom holds must be strongly connected by correct channels.
We defer its proof to \tr{\ref{sec:app-lattice-lower}}{\connected}.
\begin{lemma}
\label{lemma:lattice-agreement-lower-1}
Let $f$ be a failure pattern and $T \subseteq \PP$.
If some algorithm $\ALG$ is an $(f,T)$-obstruction-free implementation
of lattice agreement, then $T$ is strongly connected in $\RG{\GG}{f}$.
\end{lemma}

\begin{figure}[t]
    \centering
    \includegraphics[width=0.4\linewidth,trim=600 300 600 300,clip]{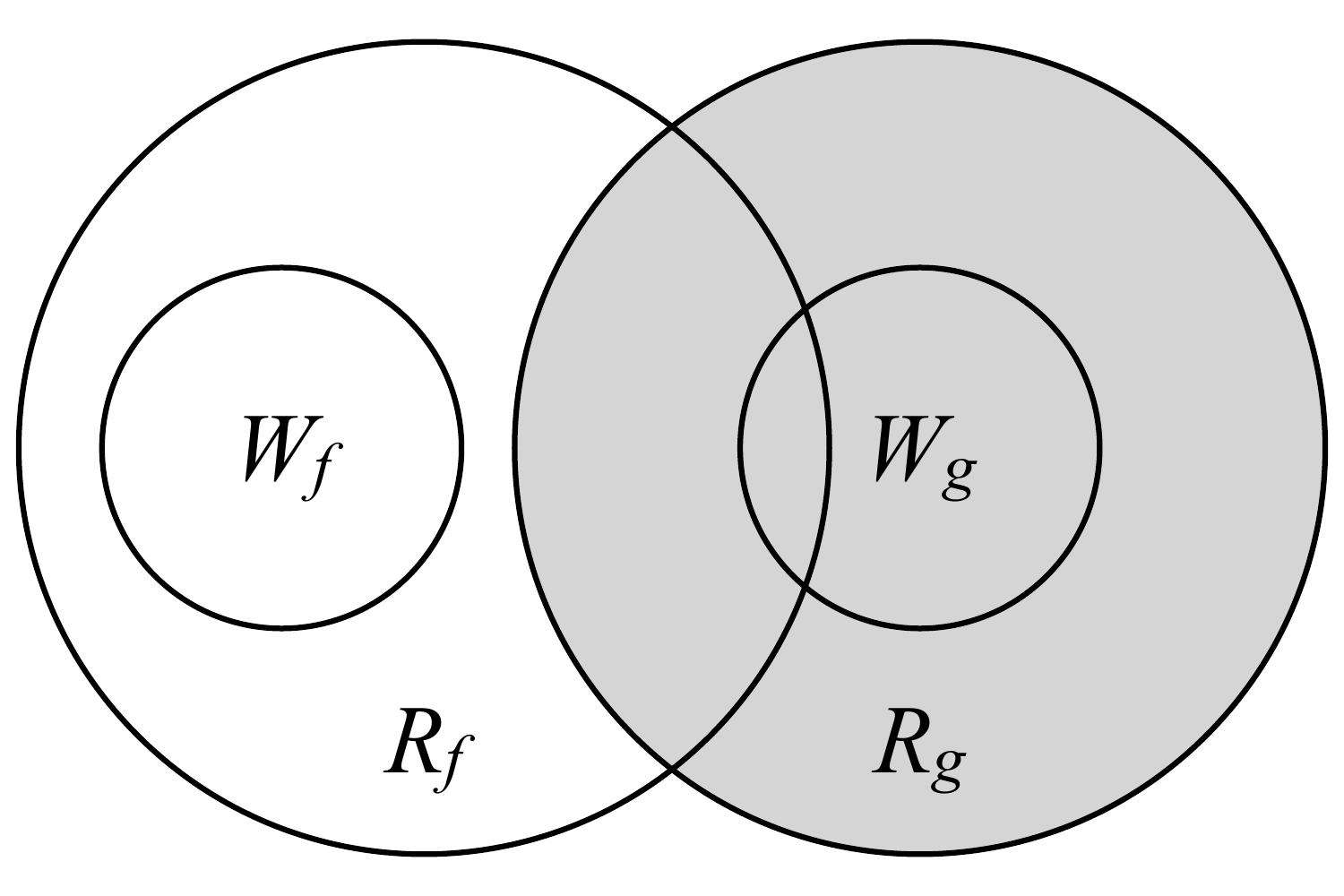}
    \caption{Illustration of the sets $W_k$ and $R_k$ for $k\in\{f,g\}$.}
    \Description{Illustration of the sets $W_k$ and $R_k$ for $k\in\{f,g\}$.}
    \label{fig:lattice-agreement-lower-1}
\end{figure}

\begin{proof}[Proof of Theorem~\ref{thm:reg-lower} for lattice agreement.]
Let $\ALG$ be an $(\FS,\tau)$-obstruction-free implementation of lattice
agreement.
For a failure pattern $f\in\FS$, Lemma~\ref{lemma:lattice-agreement-lower-1} 
implies that the set $\tau(f)$ of processes where obstruction-freedom holds 
must be transitively connected via correct channels.
Let then $W_f$ be the strongly connected component of $\GG\setminus f$
containing $\tau(f)$ and $R_f$ be the set of processes that can reach $W_f$ in 
$\GG\setminus f$, including $W_f$ itself.
Finally, let $\RR=\{R_f \mid f\in\FS\}$ and $\WW=\{W_f \mid f\in\FS\}$.

We show that $(\FS,\RR,\WW)$ is a generalized quorum system.
Assume by contradiction that this is not the case.
Note that for every $f\in\FS$ we have that $W_f$ is $f$-available and
$W_f$ is $f$-reachable from $R_f$.
Thus, $(\FS,\RR,\WW)$ satisfies {\sf Availability}.
Since by assumption $(\FS,\RR,\WW)$ is not a generalized quorum system,
then it must fail to satisfy {\sf Consistency}.
Hence, there exist $f,g\in\FS$ such that $W_f\cap R_g=\emptyset$.
Refer to Figure~\ref{fig:lattice-agreement-lower-1} for a visual depiction.

\setcounter{myclaim}{0}
\begin{myclaim}
\label{claim:lattice-lower-1}
For $k\in\{f,g\}$, $R_k$ is unreachable from $\PP\setminus R_k$ in $\GG\setminus k$.
\end{myclaim}
\begin{myclaim}
\label{claim:lattice-lower-2}
For $k\in\{f,g\}$, $R_k\setminus W_k$ is unreachable from $W_k$ in $\GG\setminus k$.
\end{myclaim}

Let $\mathcal{L}$ be a semi-lattice with partial order $\leq$ such that 
$x_1, x_2 \in \mathcal{L}$, $x_1 \nleq x_2$ and $x_2 \nleq x_1$.
Let $\alpha_1$ be a fair execution of $\ALG$ 
where the processes and channels in $f$ fail at the beginning, 
a process $p_1\in \tau(f)$ invokes $\propose(x_1)$,
and no other operation is invoked in $\alpha_1$. 
Because $p_1\in \tau(f)$ and $\ALG$ is $(f,\tau(f))$-obstruction-free, 
the $\propose(x_1)$ operation must eventually terminate with a 
return value $y_1$.
By {\sf Downward validity}, $x_1 \leq y_1$.
By {\sf Upward validity}, since no other $\propose()$ was invoked, $y_1 \leq x_1$.
Thus, $y_1 = x_1$. 
Let $\alpha_2$ be the prefix of $\alpha_1$
ending with the response to $\propose()$.
By Claim~\ref{claim:lattice-lower-1}, $R_f$ is unreachable from 
$\PP\setminus R_f$, and thus, 
the actions by processes in $R_f$ do not depend
on those by processes in $\PP\setminus R_f$. 
Then $\alpha_3=\alpha_2|_{R_f}$, the projection of $\alpha_2$
to processes in $R_f$, is an execution of $\ALG$.
By Claim~\ref{claim:lattice-lower-2}, $R_f\setminus W_f$ is unreachable
from $W_f$, and thus, the actions by processes in $R_f\setminus W_f$
do not depend on those by processes in $W_f$. 
Then $\alpha=\alpha_3|_{R_f\setminus W_f}\alpha_3|_{W_f}$
is an execution of $\ALG$.
Finally, $\alpha_3|_{R_f\setminus W_f}$ 
contains no $\propose$ invocations (only startup steps).

Let $\beta_1$ be a fair execution of $\ALG$ that starts
with all the actions from $\alpha_3|_{R_f\setminus W_f}$,
followed by the failure of all processes and channels in $g$,
followed by a $\propose(x_2)$ invocation by a process $p_2\in \tau(g)$.
Because $p_2\in \tau(g)$ and $\ALG$ is $(g,\tau(g))$-obstruction-free, 
the $\propose(x_2)$ operation must eventually terminate
with a return value $y_2$.
By {\sf Downward validity}, $x_2 \leq y_2$.
By {\sf Upward validity}, since no other $\propose()$ was invoked, $y_2 \leq x_2$.
Thus, $y_2=x_2$.
Let $\beta_2$ be the prefix of $\beta_1$ ending with the response to
$\propose()$ and let $\delta$ be the suffix of $\beta_2$ such that
$\beta_2=\alpha_3|_{R_f\setminus W_f}\delta$.
By Claim~\ref{claim:lattice-lower-1}, $R_g$ is unreachable from 
$\PP\setminus R_g$, and 
thus, the actions in $\delta$ by processes in $R_g$ do not depend
on those by processes in $\PP\setminus R_g$.
Then, $\beta=\alpha_3|_{R_f\setminus W_f}\delta|_{R_g}$
is an execution of $\ALG$.

Consider the execution $\sigma=\alpha_3|_{R_f\setminus W_f}\alpha_3|_{W_f}\delta|_{R_g}$ 
where no process or channel fails. We have:
\begin{gather*}
\sigma|_{R_f\setminus R_g}=
(\alpha_3|_{R_f\setminus W_f}\alpha_3|_{W_f}\delta|_{R_g})|_{R_f\setminus R_g}
={}
  \\
  (\alpha_3|_{R_f\setminus W_f}\alpha_3|_{W_f})|_{R_f\setminus R_g}
=
\alpha|_{R_f\setminus R_g}.
\end{gather*}
Thus, $\sigma$ is indistinguishable from $\alpha$ 
to the processes in $R_f\setminus R_g$.
Also, because $W_f\cap R_g=\emptyset$, we have:
$$
\sigma|_{R_g} =
(\alpha_3|_{R_f\setminus W_f}\alpha_3|_{W_f}\delta|_{R_g})|_{R_g} =
(\alpha_3|_{R_f\setminus W_f}\delta|_{R_g})|_{R_g} =
\beta|_{R_g}.
$$
Thus, $\sigma$ is indistinguishable from $\beta$ to the processes in $R_g$.
Finally, $\sigma|_{\PP\setminus(R_f\cup R_g)}=\varepsilon$. Thus, for every process, 
$\sigma$ is indistinguishable to this process from some execution of $\ALG$. 
Furthermore, each message received by a process in $\sigma$ has previously been sent 
by another process. Therefore, $\sigma$ is an execution of $\ALG$. However, in this 
execution $p_1$ decides $x_1$, $p_2$ decides $x_2$, and $x_1, x_2$ are incomparable 
in $\mathcal{L}$. This contradicts the {\sf Comparability} property of lattice agreement.
The contradiction derives from assuming that $(\FS,\RR,\WW)$ is not a generalized quorum 
system, so the required follows. Finally, for each $f\in\FS$ we have $\tau(f)\subseteq 
W_f\subseteq \ter{f}$.
\end{proof}


\section{Tight Bound for Consensus}
\label{sec:consensus}

We now move from the asynchronous model we have used so far to the {\em
partially synchronous model}~\cite{dls}. We show that, in this model, the
existence of a generalized quorum system is also a tight bound on the process
and channel failures that can be tolerated by any implementation of consensus. The
partially synchronous model assumes the existence of a \emph{global
stabilization time} ($\GST$) and a bound $\delta$ such that after $\GST$,
every message sent by a correct process on a correct channel is received within
$\delta$ time units of its transmission. Messages sent before $\GST$ may
experience arbitrary delays. Additionally, the model assumes that processes have
clocks that may drift unboundedly before $\GST$, but do not drift
thereafter. Both $\GST$ and $\delta$ are a priori unknown.

Let $\mathsf{Value}$ be an arbitrary domain of values. The consensus object
provides a single operation $\propose(x)$, $x \in \mathsf{Value}$, which returns
a value in $\mathsf{Value}$. The object has to satisfy the standard safety
properties: all terminating $\opropose()$ invocations must return the same value
({\sf Agreement}); and $\opropose()$ can only return a value passed to some
$\opropose()$ invocation ({\sf Validity}). The following theorems are analogs of
Theorems~\ref{thm:reg-upper} and~\ref{thm:reg-lower} for consensus.
\begin{theorem}
\label{thm:consensus-upper}
Let $(\FS,\RR,\WW)$ be a generalized quorum system and
$\tau : \FS \rightarrow 2^{\PP}$ be the termination mapping such that for each
$f\in\FS$, $\tau(f)=\ter{f}$.
Then there exists an $(\FS,\tau)$-wait-free implementation of consensus.
\end{theorem}
\begin{theorem}
\label{thm:consensus-lower}
Let $\mathcal{F}$ be a fail-prone system and $\tau : \FS \rightarrow 2^{\PP}$ be
the termination mapping such that for each $f \in \FS$, $\tau(f)\neq\emptyset$.
If there exists an $(\FS,\tau)$-obstruction-free implementation of consensus,
then there exist $\RR$ and $\WW$ such that $(\FS,\RR,\WW)$ is a
generalized quorum system. Moreover, for each $f\in\FS$, we have $\tau(f)\subseteq \ter{f}$.
\end{theorem}

We first present a consensus protocol validating
Theorem~\ref{thm:consensus-upper}. To this end, we fix a generalized quorum system
$(\FS,\RR,\WW)$ and a termination mapping $\tau : \FS \rightarrow 2^{\PP}$ such
that $\tau(f)=\ter{f}$ holds for each $f \in \FS$. As in
\S\ref{sec:upper-bound}, we assume without loss of generality that the
connectivity relation of the graph $\GG\setminus f$ is transitive for each
$f \in \FS$. Our protocol for consensus is shown in
Figure~\ref{code:css}.

\begin{figure}[t]
\removelatexerror%
	\begin{algorithm*}[H]
		\DontPrintSemicolon
		\SetAlgoNoLine
		\setcounter{AlgoLine}{0}

		\assign{\view,\cview}{0,0}\;
		\assign{\val,\pval}{\bot,\bot}\;
		$\phase\in\{\entered,\proposed,
			\accepted,\decided\}$\;

		\smallskip\smallskip

		\Function{$\opropose(x)$}{
				$\assign{\pval}{x}$\;\label{css:propose:myval}
			\textbf{async wait until} $\phase = \decided$\;\label{css:propose:wait}
			\textbf{return}\xspace $\val$\;
		}

		\smallskip\smallskip

		\SubAlgo{\mbox{\onreceive $\{\OB(\view, v_j, x_j) \mid p_j\in R\}$
                  {\bf from some} $R \in \RR$}}{\label{css:ob}
			\precond $\phase = \entered$\;\label{css:ob:pre}
			\uIf{$\forall p_j \in R .\ x_j = \bot$}{\label{css:ob:if}
				\lIf{$\pval = \bot$}{\textbf{return}}\label{css:ob:if:no-val}
				\send $\TA(\view, \pval)$ \ToAll\;\label{css:ob:if:send}
			}\Else{\label{css:ob:else}
				{\bf let $p_j \in R$ be such that $x_j \neq \bot
                                  \land {}$}\linebreak 
					\phantom \ \ \ $(\forall p_k \in R.\ x_k \neq \bot \implies v_k \leq v_j)$\;\label{css:ob:else:max}
				\send $\TA(\view, x_j)$ \ToAll\;\label{css:ob:else:send}
			}
			$\assign{\phase}{\proposed}$\;\label{css:ob:phase}
		}

		\smallskip\smallskip

	\SubAlgo{\onreceive $\TA(\view, x)$}{\label{css:ta}
		\precond $\phase \in \{\entered,\proposed\}$\;
		$\assign{\val}{x}$\;
		$\assign{\cview}{\view}$\;
		\send $\TB(\view, x)$ \ToAll\;
		$\assign{\phase}{\accepted}$\;
	}

	\smallskip\smallskip

	\SubAlgo{\mbox{\onreceive $\{\TB(\view, x) \mid p_j\in W\}$
          {\bf from some} $W\in\WW$}}{\label{css:tb}
		$\assign{\val}{x}$\;
		$\assign{\cview}{\view}$\;
		$\assign{\phase}{\decided}$\;
	}

	\smallskip\smallskip
	
	\SubAlgo{{\bf on startup or when the timer $\viewtimer$ expires}}{\label{css:timer}
		$\assign{\view}{\view + 1}$\;\label{css:timer:view}
		$\starttimer(\viewtimer, \view \cdot C)$\;\label{css:timer:timer}
		\send $\OB(\view, \cview, \val)$ \To $\leader(\view)$\;\label{css:timer:send}
		$\assign{\phase}{\entered}$\;\label{css:timer:phase}
	}

	\end{algorithm*}
	\caption{The consensus protocol at a process $p_i$.}
	\Description{The consensus protocol at a process $p_i$.}
	\label{code:css}
\end{figure}

\paragraph{Consensus vs registers under channel failures.}  Interestingly,
solving consensus under process and channel failures is simpler than
implementing registers. The main challenge we had to deal with when implementing
registers was determining whether the information received by a process is up to
date (\S\ref{sec:upper-bound}). This is particularly difficult in the
asynchronous model with unidirectional connectivity, where processes cannot rely
on bidirectional exchanges to confirm the freshness of the information.  In the
partially synchronous model, however, processes can exploit the eventual
timeliness of the network to determine freshness.  Technically, this is done
using a {\em view synchronizer}~\cite{cogsworth,tenderbake}, explained
next. Then the connectivity stipulated by a generalized quorum system is
sufficient to implement $(\FS,\tau)$-wait-free consensus using an algorithm
similar to Paxos~\cite{paxos}.

\paragraph{View synchronization.} The consensus protocol works in a
succession of views, each with a designated leader process
$\leader(v)=p_{((v-1) \!\! \mod n)+1}$. Thus, the role of the leader rotates
round-robin among the processes. The current view is tracked in a variable
$\view$. The protocol synchronizes views among processes via growing
timeouts~\cite{cogsworth,tenderbake}. Namely, each process spends the time
$v \cdot C$ in view $v$, where $C$ is an arbitrary positive constant. To ensure
this, upon entering a view $v$, the process sets a timer $\viewtimer$ for the
duration $v \cdot C$ (line~\ref{css:timer:timer}). When the timer expires, the
process increments its $\view$ (line~\ref{css:timer:view}). Hence, the time
spent by a process in each view grows monotonically as views increase. Even
though processes do not communicate to synchronize their views, this simple
mechanism ensures that all correct processes overlap for an arbitrarily long
time in all but finitely many views.

\begin{proposition}
\label{proposition:sync}
Let $d$ be an arbitrary positive value. There exists a view $\VV$ such that for
every view $v\geq\VV$, all correct processes overlap in $v$ for at least $d$.
\end{proposition}

\paragraph{Protocol operation.} A process stores its initial proposal in
$\pval$ (line~\ref{css:propose:myval}), the last proposal it accepted in $\val$
and the view in which this 
happened in $\cview$. A variable $\phase$ tracks the progress of the process
through the different phases of the protocol. When a process enters a view $v$,
it sends the information about its last accepted value to $\leader(v)$ in a
$\OB$ message (line~\ref{css:timer:send}). This message is analogous to the
$\OB$ message of Paxos; there is no analog of a $\OA$ message,
because leader election is controlled by the synchronizer.

A leader waits until it receives $\OB$ messages from every member of some read
quorum corresponding to its view (line~\ref{css:ob}); 
messages from lower views (considered out of date) are ignored.
Based on these messages, the leader
computes its proposal similarly to Paxos. If some process has previously
accepted a value, the leader picks the one accepted in the maximal view. If
there is no such value and $\opropose()$ has already been invoked at the leader,
the leader picks its own value. Otherwise, it skips its turn.

A process waits until it receives a $\TA$ message from the leader of its view
(line~\ref{css:ta}) and accepts the proposal by updating its $\val$ and
$\cview$. It then notifies every process about this through a $\TB$
message. Finally, when a process receives matching $\TB$ messages from every
member of some write quorum (line~\ref{css:tb}), it knows that a decision has
been reached, and it sets $\phase=\decided$. If there is an ongoing $\opropose()$
invocation, this validates the condition at line~\ref{css:propose:wait}, and the
process returns the decision to the caller.

\begin{proof}[Proof of Theorem~\ref{thm:consensus-upper}.]
  It is easy to see that the protocol satisfies {\sf Validity}. The proof of
  {\sf Agreement} is virtually identical to that of Paxos, relying on the {\sf
    Consistency} property of the generalized quorum system. We now prove that
  the protocol is $(\FS,\tau)$-wait-free.

	Fix a failure pattern $f\in\FS$ and a process $p\in\tau(f)$ that invokes 
	$\opropose(x)$ for some $x$ at a time $t'$.
	By {\sf Availability}, there exist $W\in\WW$ and $R\in\RR$ such that $W$ is 
	$f$-available, and $W$ is $f$-reachable from $R$. By
	Proposition~\ref{proposition:f-avail-scc},
	$W\subseteq U_f$. Then since $p\in\tau(f)=U_f$, $p$ is strongly 
	connected to all processes in $W$ via channels correct under $f$. Thus,
        the following hold after $\GST$: 
	{\em (i)} $R$ can reach $p$ through timely channels; and {\em (ii)}
	$p$ can exchange messages with every member of $W$ via timely channels.

	By Proposition~\ref{proposition:sync} and since leaders rotate round-robin,
	there exists a view $v$ led by $p$ such that all correct processes
	enter $v$ after $\max(\GST,t')$ and overlap in this view for more than $3\delta$.
	We now show that this overlap is sufficient for $p$ to reach a decision.
	Let $t$ be the earliest time by which every correct process has
        entered $v$. Then no correct process leaves $v$ until after $t+3\delta$.
	When a process enters $v$, it sends a $\OB$ message to $p$
	(line~\ref{css:timer:send}). By {\em (i)}, $p$ is guaranteed to receive $\OB$
	messages for view $v$ from every member of $R$ by the time $t+\delta$, 
	thereby validating the guard at line~\ref{css:ob}.
	As a result, by this time $p$ will send its proposal in a $\TA$ message while still in $v$.
	By {\em (ii)}, each process in $W$ will receive this message no later than $t+2\delta$, 
	and respond with a $\TB$ message that will reach $p$ (line~\ref{css:ta}).  
	Thus, by the time $t+3\delta$, $p$ is guaranteed to collect $\TB$ messages for view $v$ 
	from every member of $W$ (line~\ref{css:tb}). After this $p$ sets $\phase=\decided$, thus
	satisfying the guard at line~\ref{css:propose:wait} and deciding.
\end{proof}

\paragraph{Lower bound.}  We now argue that Theorem~\ref{thm:consensus-lower}
holds. First, note that the proof of the lower bound for lattice agreement
(\S\ref{sec:lower-bound}) remains valid in the partially synchronous model:
since the execution $\sigma$ constructed in the proof is finite, it is also
valid under partial synchrony where all its actions occur before $\GST$. Just like
under asynchrony, the lower bound for lattice agreement directly implies the one
for registers under partial
synchrony~\cite{hagit-snapshots,hagit-lattice-agreement}. Finally, since
consensus can be used to implement a MWMR atomic register, the lower bound for
consensus follows from the bound for registers.


\section{Related Work}

Research on tolerating channel failures has primarily focused on consensus
solvability in synchronous systems, where an adversary can disconnect channels
in every round but cannot crash processes
~\cite{time-not-healer,santoro-widmayer2,COULOUMA201580,message-adv,CG13,SWK09,NSW19}.
The seminal paper by Dolev~\cite{dolev-strikes-again} and subsequent work
(see~\cite{topology-consensus} for survey) explored this problem in general
networks as a function of the network topology, process failure models, and
synchrony constraints.  However, this work considers only consensus formulations
requiring termination at all correct processes, which in its turn requires them
to be reliably connected.  In contrast, we consider more general liveness
conditions where termination is only required at specific subsets of correct
processes. Our results demonstrate that this relaxation leads to a much richer
characterization of connectivity, which notably does not require all correct
processes to be bidirectionally connected.

Early models, such as \emph{send omission}~\cite{hadzilacos-send-omit} and
\emph{generalized omission}~\cite{perry-toueg-omit}, extended crash failures by allowing 
processes to fail to send or receive messages. These models were shown to be
computationally equivalent to crash failures in both synchronous~\cite{neiger-toueg} and
asynchronous~\cite{coan} systems. However, they are overly restrictive as they classify 
any non-crashed process with unreliable connectivity as faulty. In contrast, our approach 
allows correct processes to have unreliable connectivity.
Santoro and Widmayer introduced the \emph{mobile omission} failure model~\cite{time-not-healer},
which decouples message loss from process failures. They demonstrated that solving consensus
in a synchronous round-based system without process failures requires 
the communication graph to contain a strongly connected component in every 
round~\cite{time-not-healer,santoro-widmayer2}. In contrast, our lower bounds show that 
implementing consensus -- or even a register -- in a partially synchronous system necessitates 
connectivity constraints that hold throughout the entire execution.

Failure detectors and consensus have been shown to be implementable 
in the presence of network partitions~\cite{friedman1,friedman2,friemdan-podc-ba,aguilera-heartbeat}
provided a majority of correct processes are strongly connected and can eventually 
communicate in a timely fashion. In contrast, we show that much weaker connectivity
constraints, captured via generalized quorum systems, are necessary and
sufficient for implementing both register and consensus.
Aguilera et al.~\cite{aguilera-podc03-journal} and subsequent 
work~\cite{aguilera-podc04,dahlia-t-accessible,dahila-t-moving-source,antonio-intermittent-star}
studied the implementation of 
$\Omega$ -- the weakest for consensus~\cite{CT96-weakest} --
under various weak models of synchrony, link reliability, and connectivity.
This work, however, mainly focused on identifying minimal timeliness requirements
sufficient for implementing $\Omega$ while assuming at least fair-lossy 
connectivity between each pair of processes. Given that reliable channels can be 
implemented on top of fair-lossy ones, our results imply that these
connectivity conditions are too strong. 
Furthermore, while the weak connectivity conditions of system $S$ introduced
in~\cite{aguilera-podc03-journal} were shown to be sufficient for
implementing $\Omega$, they were later proven insufficient for consensus~\cite{opodis}.

In our previous work~\cite{opodis}, we considered systems with process crashes and
channel disconnections for $n=2k+1$, where any $k$ processes can fail. We
proved that a majority of reliably connected correct processes is necessary for
implementing registers or consensus. This work, however, does not address
solvability under arbitrary fail-prone systems, such as those with
$k < \lfloor\frac{n-1}{2}\rfloor$ or those not based on failure
thresholds~\cite{quorum-systems-naor,bqs}.

Alquraan et al.~\cite{osdi-partitions} presented a study of system failures due to
faulty channels, which we already mentioned in \S\ref{sec:intro}. This work
highlights the practical importance of designing provably correct systems
that explicitly account for channel failures. Follow-up
work~\cite{osdi-partitions2,omnipaxos} proposed practical systems for tolerating
channel failures, but did not investigate optimal fault-tolerance assumptions.


\begin{acks}
This work was partially supported by the projects BYZANTIUM, DECO and PRODIGY
funded by MCIN/AEI, and REDONDA funded by the CHIST-ERA network.
We thank Petr Kuznetsov for helpful comments and discussions.
\end{acks}

\pagebreak

\clearpage

\balance
\bibliographystyle{ACM-Reference-Format}
\bibliography{references}

\iflong
  \clearpage
  \appendix  
  \section{Safety Specifications}
\label{sec:app-safety-defs}

\paragraph{Auxiliary definitions.} 
An operation $\mathit{op'}$ {\em follows}\ an operation $\mathit{op}$, denoted
$\mathit{op} \rightarrow \mathit{op'}$, if $\mathit{op'}$ is invoked after
$\mathit{op}$ returns; $\mathit{op'}$ is {\em concurrent}\/ with $\mathit{op}$
if neither $\mathit{op} \rightarrow \mathit{op'}$ nor
$\mathit{op'} \rightarrow \mathit{op}$. An execution is {\em sequential} if no
operations in it are concurrent with each other.  An execution $\sigma$ of an
object (e.g., register or snapshot) is
\emph{linearizable}~\cite{linearizability} if there exists a set of responses
$R$ and a sequence $\pi = \mathit{op}_1,\mathit{op}_2,\dots$ of all complete
operations in $\sigma$ and some subset of incomplete operations paired with
responses in $R$, such that $\pi$ is a correct sequential execution and
satisfies $\mathit{op}_i \rightarrow \mathit{op}_j \implies i < j$.  The
implementation of an object is \emph{atomic}\/ if all its executions are
linearizable.


\paragraph{MWMR Atomic Registers.}
Let $\Val$ be an arbitrary domain of values. A \emph{multi-writer multi-reader
  atomic register} supports two operations: $\owrite(x)$ stores a value
$x\in \Val$ and returns $\ack$; and $\oread$ retrieves the current value from the
register and returns it. A sequential execution of a register is \emph{correct}
if every $\oread$ operation returns the value written by the most recent
preceding $\owrite$ operation.

\paragraph{SWMR Atomic Snapshots.}
A \emph{single-writer multi-reader atomic snapshot object} consists of
\emph{segments}, where each segment holds a value in $\Val$. In the
single-writer case, each process is assigned a unique segment, which only that
process writes to. All processes can read from all segments. The interface
supports two operations: $\owrite(x)$ allows a process
to store the value $x\in\Val$ in its segment and returns $\ack$; and $\oread$
retrieves the values of all segments as a \emph{vector} of elements in $\Val$.
A sequential execution of a snapshot object is \emph{correct} if every $\oread$
operation returns a vector that reflects the values written by the most recent
preceding $\owrite$ operations on each segment.


  \section{Proof of Linearizability for the Protocol in Figure~\ref{code:reg}}
\label{sec:reg-safety}

We now show that the protocol in 
Figure~\ref{code:reg} is linearizable, thereby
completing the proof of Theorem~\ref{thm:reg-upper}.
For simplicity, we only consider executions of the algorithm
where all operations complete.
To each execution of the algorithm, we associate:

\begin{itemize}
    \item a set $V(\sigma)$ consisting of the operations in $\sigma$,
        i.e., $\oread$s and $\owrite$s; and
    \item a relation $\rt(\sigma)$, defined as follows:
        for all $o_1,o_2\in\sigma$, $(o_1,o_2)\in\rt(\sigma)$
        if and only if $o_1$ completes before $o_2$ is invoked.
\end{itemize}

We denote the $\oread$ operations in $\sigma$ by $R(\sigma)$ and 
the $\owrite$ operations in $\sigma$ by $W(\sigma)$. A 
\emph{dependency graph} of $\sigma$ is a tuple
$G=(V(\sigma),\rt(\sigma),\WR,\ww,\rw)$, where the relations
$\WR,\ww,\rw\subseteq V(\sigma)\times V(\sigma)$ are such that:

\begin{enumerate}
    \item 
        \emph{(i)} if $(o_1,o_2)\in\WR$, then $o_1\in W(\sigma)$ 
            and $o_2\in R(\sigma)$;\\
        \emph{(ii)} for all $w_1,w_2,r$ $\in \sigma$ such that
            $(w_1,r)\in\WR$ and $(w_2,r)\in\WR$, we have $w_1=w_2$;\\
        \emph{(iii)} for all $(w,r)\in\WR$ we have $\val(w)=\val(r)$; and\\
        \emph{(iv)} if there is no $w\in W(\sigma)$ such that $(w,r)\in\WR$,
            then $r$ returns $0$;
    \item $\ww$ is a total order over $W(\sigma)$; and
    \item $\rw = \{(r,w) \mid \exists w'.\ (w',r)\in\WR \wedge (w',w)\in\ww\}\ \cup$\\
    \phantom\ \ \ \ \ \ \ \ \ \,$\{(r,w) \mid r\in R(\sigma) \wedge w\in W(\sigma) \wedge \neg\exists w'.\ (w',r)\in\WR\}$.
\end{enumerate}

To prove linearizability we rely on the following theorem\footnote{Atul
  Adya. 1999. Weak Consistency: A Generalized Theory and Optimistic
  Implementations for Distributed Transactions. Ph.D thesis, MIT, Technical
  Report MIT/LCS/TR-786.
}

\begin{theorem}
    \label{thm:linearizability}
    An execution $\sigma$ is linearizable if and only if there exists
    $\WR$, $\ww$ and $\rw$ such that $G=(V(\sigma),\rt(\sigma),\WR,\ww,\rw)$ 
    is an acyclic dependency graph.
\end{theorem}

We now prove that every execution $\sigma$ of the protocol is linearizable.
Fix one such execution $\sigma$. Our strategy is to find witnesses for 
$\WR$, $\ww$ and $\rw$ that validate the conditions of Theorem~\ref{thm:linearizability}.
To this end, consider the function $\tau:\sigma\rightarrow\mathbb{N}\times\mathbb{N}$ that
maps each operation in $\sigma$ to a version as follows:

\begin{itemize}
    \item for a $\oread$ $r$, $\tau(r)$ is the version of $s'$ at
      line~\ref{code:reg:read:t} in Figure~\ref{code:reg}; and
    \item for a $\owrite$ $w$, $\tau(w)$ is the version $t$ at
      line~\ref{code:reg:t} in Figure~\ref{code:reg}.
\end{itemize}

We then define the required witnesses as follows:

\begin{itemize}
    \item $(w,r)\in\WR$ if and only if $w\in W(\sigma)$, $r\in R(\sigma)$ and $\tau(w)=\tau(r)$;
    \item $(w,w')\in\ww$ if and only if $w,w'\in W(\sigma)$ and $\tau(w)<\tau(w')$; and
    \item $\rw$ is derived from $\WR$ and $\ww$ as per the dependency graph definition.
\end{itemize}

Our proof relies on the next proposition.
We omit its easy proof which follows from the {\sf Validity} property
of the quorum access functions:

\begin{proposition}
    \label{prop:lin-aux}
    The following hold:
    \begin{enumerate}
        \item For every $w_1,w_2\in W(\sigma)$, $\tau(w_1)=\tau(w_2)$ implies $w_1=w_2$.
        \item For every $w\in W(\sigma)$, $\tau(w)>(0,0)$.
        \item For every $r\in R(\sigma)$, either $\tau(r)=(0,0)$ or there exists
            $w\in W(\sigma)$ such that $\tau(r)=\tau(w)$.
        \item For every $r\in R(\sigma)$ and $w\in W(\sigma)$, $\tau(r)=\tau(w)$ 
            implies $\val(r)=\val(w)$.
    \end{enumerate}
\end{proposition}

Our proof also relies on the following auxiliary lemma:

\begin{lemma}
    \label{lemma:lin-aux}
    The following hold:
    \begin{enumerate}
        \item For all $r,w\in V(\sigma)$, if $(r,w)\in\rw$ then $\tau(r)<\tau(w)$.
        \item For all $o_1,o_2\in V(\sigma)$, if $(o_1,o_2)\in\rt$, then $\tau(o_1)\leq\tau(o_2)$.
            Moreover, if $o_2$ is a $\owrite$, then $\tau(o_1)<\tau(o_2)$.
    \end{enumerate}
\end{lemma}

\begin{proof}
    \begin{enumerate}
        \item Let $r,w\in V(\sigma)$ be such that $(r,w)\in\rw$. There are two cases:
            \begin{itemize}
                \item Suppose that for some $w'$ we have $(w',r)\in\WR$ and
                    $(w',w)\in\ww$. The definition of $\WR$ implies that 
                    $\tau(r)=\tau(w')$, and the definition of $\ww$ implies
                    that $\tau(w')<\tau(w)$. Then $\tau(r)<\tau(w)$.
                \item Suppose now that $\neg \exists w'.\ (w',r)\in\WR$.
                    We show that $\tau(r)=(0,0)$. Indeed, if $\tau(r)\neq(0,0)$,
                    then by Proposition~\ref{prop:lin-aux}(3), there exists 
                    $w\in V(\sigma)$ such that $\tau(r)=\tau(w)$.
                    But then $(w,r)\in\WR$, contradicting the assumption that
                    there is no such write. At the same time,
                    Proposition~\ref{prop:lin-aux}(2) implies that 
                    $\tau(w)>(0,0)$. Then $\tau(r)<\tau(w)$.
            \end{itemize}
        \item Let $o_1,o_2\in V(\sigma)$ be such that $(o_1,o_2)\in\rt$.
            Let $u$ be the update function computed during $o_1$'s invocation at 
            lines~\ref{code:reg:write:upd} (if $o_1$ is a $\owrite$) or
            line~\ref{code:reg:read:upd} (if $o_1$ is a $\oread$).
            The definitions of $u$ imply that right after a process applies
            it to its state it has $\ts\geq \tau(o_1)$.
            
            Suppose now that $o_2$ is invoked at a process $p$.
            Let $S=\{s_i \mid i=1..k\}$ be the states returned by the 
            corresponding $\readquorum()$ invocation
            (lines~\ref{code:reg:write:read-quorum}~and~\ref{code:reg:read:read-quorum}).
            The {\sf Validity} property of the quorum access functions
            ensures that for each $s_i\in S$
            there exists a set of previous invocations
            $\{\writequorum(u^i_j) \mid j=1..k_i\}$ such that $s_i$ is the result of applying
            the update functions in $\{u_j^i \mid j=1..k_i\}$ to the initial state in some
            order. Since the $\readquorum()$ invocation happens after
            $\writequorum(u)$ has completed,
            by the {\sf Real-time ordering} property there is at least
            one state $s_r$ to which $u$ has been applied:
            $u = u_j^r$ for some $j$.
            Because the update functions passed by the protocol to
            $\writequorum()$ never decrease $\ts$, we then have $s_r.\ts\geq\tau(o_1)$.
            There are two cases:
            \begin{itemize}
                \item Suppose $o_2$ is a $\owrite$. 
                    Because $\tau(o_2)$ is greater than the maximum $\ts$ value among 
                    the states in $S$ (lines~\ref{code:reg:t'}-\ref{code:reg:t}),
                    we have $\tau(o_1)<\tau(o_2)$.
                \item Suppose $o_2$ is a $\oread$. 
                    Because $\tau(o_2)$ is the maximum $\ts$ value among 
                    the states in $S$ (line~\ref{code:reg:read:t}),
                    we have $\tau(o_1)\leq\tau(o_2)$.
            \end{itemize}
    \end{enumerate}
\end{proof}

\begin{theorem}
    $G=(V(\sigma),\rt(\sigma),\WR,\ww,\rw)$ is an acyclic dependency graph.
\end{theorem}

\begin{proof}
    From Proposition~\ref{prop:lin-aux} and the definitions of $\WR$, $\ww$ and $\rw$ 
    it easily follows that $G$ is a dependency graph. 
    We now show that $G$ is acyclic. By contradiction, assume that the graph contains a 
    cycle $o_1, \dots, o_n = o_1$. Then $n > 1$. 
    By Lemma~\ref{lemma:lin-aux} and the definitions of $\tau$ and $\ww$, we must have 
    $\tau(o_1) \leq \dots \leq \tau(o_n) = \tau(o_1)$, so that $\tau(o_1) = \dots = \tau(o_n)$. 
    Furthermore, if $(o,o')$ is an edge of $G$ and $o'$ is a write, then $\tau(o) < \tau(o')$. 
    Hence, all the operations in the cycle must be reads, and thus, all the edges in the 
    cycle come from $\rt$. Then there exist reads $r_1$, $r_2$ in the cycle such that 
    $r_1$ completes before $r_2$ is invoked and $r_2$ completes before $r_1$ is invoked, 
    which is a contradiction.
\end{proof}


  \section{Proof of Lemma~\ref{lemma:lattice-agreement-lower-1}}
\label{sec:app-lattice-lower}

\begin{figure}[t]
    \centering
    \includegraphics[width=0.5\linewidth,trim=550 275 550 150,clip]{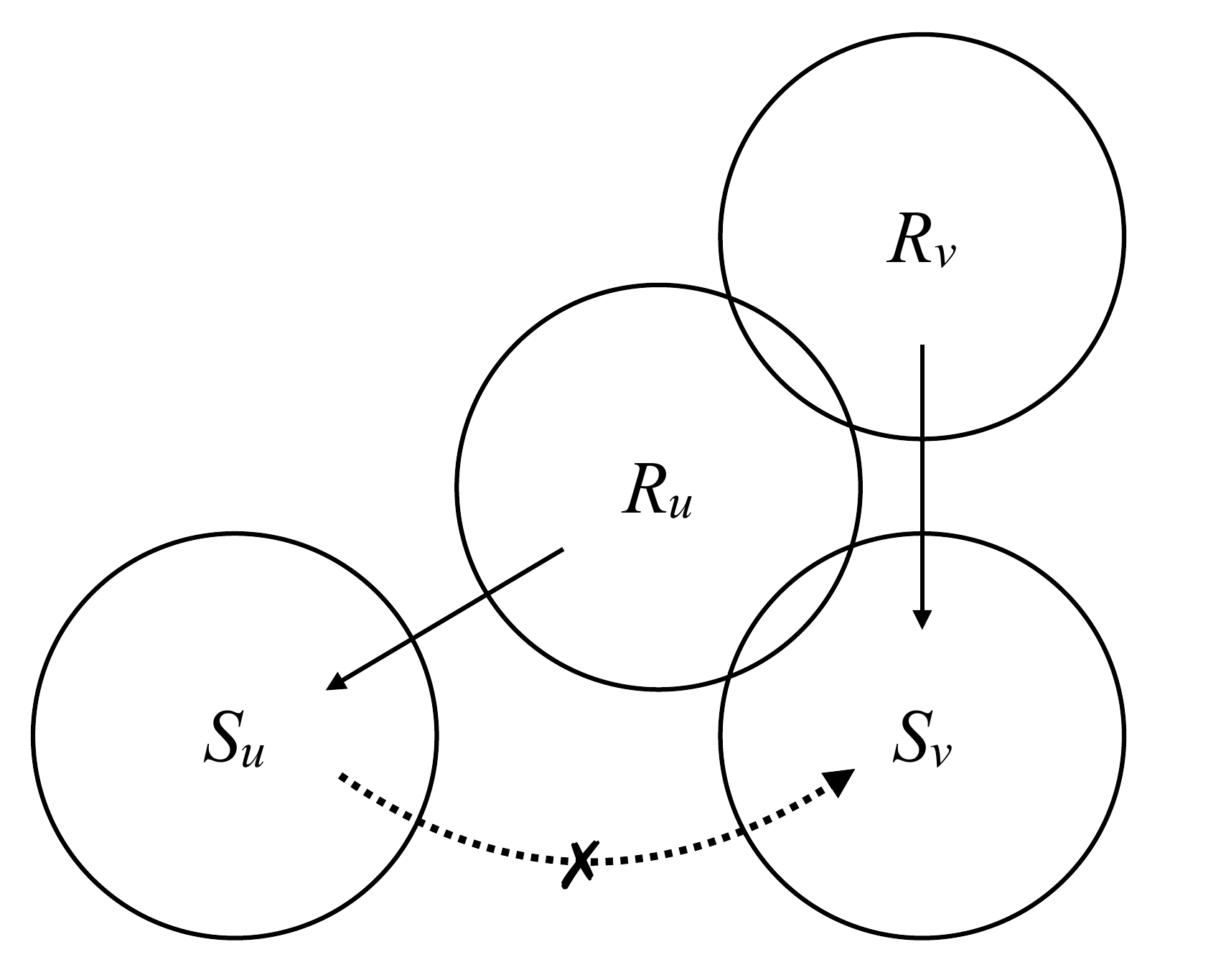}
    \caption{Illustration of the sets $R_k$ and $S_k$ for $k\in\{u,v\}$.}
    \Description{Illustration of the sets $R_k$ and $S_k$ for $k\in\{u,v\}$.}
    \label{fig:lattice-agreement-lower-2}
\end{figure}

Assume by contradiction that $\ALG$ is an $(f,T)$-obstruction-free implementation
of lattice agreement, but $T$ is not strongly connected in $\RG{\GG}{f}$.
Thus, there exist $u, v \in T$ such that there is no path from $u$ to $v$ in $\RG{\GG}{f}$,
or from $v$ to $u$. Without loss of generality, assume the former.
Let $S_u$ be the strongly connected component (SCC) of $\RG{\GG}{f}$ containing $u$,
and $S_v$ be the SCC containing $v$. By assumption, $S_u \cap S_v = \emptyset$.
Let
$R_u$ be the set of processes outside $S_u$ that can reach $u$ in $\RG{\GG}{f}$, and
$R_v$ be the set of processes outside $S_v$ that can reach $v$ in $\RG{\GG}{f}$.
Refer to Figure~\ref{fig:lattice-agreement-lower-2} for a visual depiction.

\setcounter{myclaim}{0}
\begin{myclaim}
\label{claim:lattice_lower_1}
For any $k \in \{u,v\}$, $R_k \cup S_k$ is unreachable from
$\PP \setminus (R_k \cup S_k)$ in $\RG{\GG}{f}$.
\end{myclaim}
\begin{myclaim}
\label{claim:lattice_lower_2}
For any $k \in \{u,v\}$, $R_k$ is unreachable from $S_k$ in $\RG{\GG}{f}$.
\end{myclaim}
\begin{myclaim}
\label{claim:lattice_lower_3}
$S_u \cap (R_v \cup S_v) = \emptyset$.
\end{myclaim}

\begin{proof}[Proof of Claim~\ref{claim:lattice_lower_3}.]
Assume by contradiction that $S_u \cap (R_v \cup S_v) \neq \emptyset$.
Let $w \in S_u \cap (R_v \cup S_v)$.
Since $w \in S_u$, there exists a path from $u$ to $w$.
Since $w \in R_v \cup S_v$, there exists a path from $w$ to $v$.
Concatenating these paths creates a path from $u$ to $v$,
contradicting the assumption that no such path exists.
\end{proof}

Let $\mathcal{L}$ be a semi-lattice with partial order $\leq$ such that 
$x_u, x_v \in \mathcal{L}$, $x_u \nleq x_v$ and $x_v \nleq x_u$.
Let $\alpha_1$ be a fair execution of $\ALG$ where the processes and channels
in $f$ fail at the beginning, process $u$ invokes $\propose(x_u)$, and
no other operation is invoked in $\alpha_1$.
Because $u \in T$ and $\ALG$ is $(f,T)$-obstruction-free,
the $\propose(x_u)$ operation must eventually terminate 
with a return value $y_u$. 
By {\sf Downward validity}, $x_u \leq y_u$.
By {\sf Upward validity}, since no other $\propose()$ was invoked, $y_u \leq x_u$.
Thus, $y_u = x_u$. 
Let $\alpha_2$ be the prefix of $\alpha_1$ ending with the response to $\propose()$.
By Claim~\ref{claim:lattice_lower_1}, $R_u \cup S_u$ is unreachable from
$\PP \setminus (R_u \cup S_u)$, and thus,
the actions by processes in $R_u \cup S_u$
do not depend on those by processes in $\PP \setminus (R_u \cup S_u)$.
Then $\alpha_3 = \alpha_2|_{R_u \cup S_u}$, the projection of $\alpha_2$
to actions by processes in $R_u \cup S_u$, is an execution of $\ALG$.
By Claim~\ref{claim:lattice_lower_2}, $R_u$ is unreachable from $S_u$,
and thus, the actions by processes in $R_u$ do not depend on those by processes in $S_u$.
Then $\alpha = \alpha_3|_{R_u}\alpha_3|_{S_u}$ is an execution of $\ALG$.
Finally, $\alpha_3|_{R_u}$ contains no $\propose$ invocations (only startup steps).

Let $\beta_1$ be a fair execution of $\ALG$ that starts with all the actions from $\alpha_3|_{R_u}$,
followed by the failure of all processes and channels in $f$, 
followed by a $\propose(x_v)$ invocation by the process $v$.
Because $v \in T$ and $\ALG$ is $(f,T)$-obstruction-free,
the $\propose(x_v)$ operation must eventually terminate with a return value $y_v$.
By {\sf Downward validity}, $x_v \leq y_v$.
By {\sf Upward validity}, since no other $\propose()$ was invoked, $y_v \leq x_v$.
Thus, $y_v=x_v$.
Let $\beta_2$ be the prefix of $\beta_1$ ending with the response to $\propose()$.
By Claim~\ref{claim:lattice_lower_1}, $R_u \cup S_u$ is unreachable from
$\PP \setminus (R_u \cup S_u)$. 
By Claim~\ref{claim:lattice_lower_2}, $R_u$ is unreachable from $S_u$.
Therefore, $R_u$ is unreachable from $\PP \setminus R_u$.
By Claim~\ref{claim:lattice_lower_1}, $R_v \cup S_v$ is unreachable from
$\PP \setminus (R_v \cup S_v)$. 
Therefore, $R_u \cup R_v \cup S_v$ is unreachable from $\PP \setminus (R_u \cup R_v \cup S_v)$.
Thus, the actions in $\beta_2$ by processes in $R_u \cup R_v \cup S_v$ do not depend on those by processes
in $\PP \setminus (R_u \cup R_v \cup S_v)$.
Then, $\beta = \beta_2|_{R_u \cup R_v \cup S_v}$ is an execution of $\ALG$.
Recall that $\beta_1$ starts with $\alpha_3|_{R_u}$, and hence, so does $\beta$.
Let $\delta$ be the suffix of $\beta$ such that $\beta = \alpha_3|_{R_u}\delta$.

Consider the execution $\sigma = \alpha_3|_{R_u}\delta\alpha_3|_{S_u}$ where 
no process or channel fails.
By Claim~\ref{claim:lattice_lower_3}, $S_u \cap (R_v \cup S_v) = \emptyset$, and by the
definition of $R_u$, we have $S_u \cap R_u = \emptyset$.
Hence, $S_u \cap (R_u \cup R_v \cup S_v) = \emptyset$.
Then, given that $\delta$ only contains actions by processes in $R_u \cup R_v \cup S_v$,
we get $\sigma|_{R_u \cup R_v \cup S_v} = (\alpha_3|_{R_u}\delta\alpha_3|_{S_u})|_{R_u \cup R_v \cup S_v} = 
\alpha_3|_{R_u}\delta = \beta$.
Thus, $\sigma$ is indistinguishable from $\beta$ to the processes in $R_u \cup R_v \cup S_v$.
Also,
$\sigma|_{S_u} = (\alpha_3|_{R_u}\delta\alpha_3|_{S_u})|_{S_u} = \alpha_3|_{S_u} = 
(\alpha_3|_{R_u}\alpha_3|_{S_u})|_{S_u} = \alpha|_{S_u}$.
Thus, $\sigma$ is indistinguishable from $\alpha$ to the processes in $S_u$.
Finally, $\sigma|_{\PP \setminus (R_u \cup S_u \cup R_v \cup S_v)} = \epsilon$.
Thus, for every process, $\sigma$ is indistinguishable to this process from some execution of $\ALG$. 
Furthermore, each message received by a process in $\sigma$ has previously been sent by another process. 
Therefore, $\sigma$ is an execution of $\ALG$. 
However, in this execution $u$ decides $x_u$, $v$ decides $x_v$, and $x_u, x_v$ are incomparable in $\mathcal{L}$.
This contradicts the {\sf Comparability} property of lattice agreement.
The contradiction derives from assuming that $T$ is not strongly connected in $\RG{\GG}{f}$,
so the required follows.


\fi

\end{document}